\documentclass[a4paper,onecolumn,accepted=2021-06-24]{quantumarticle}
\pdfoutput=1
 \usepackage[utf8]{inputenc}
\usepackage[english]{babel}
\usepackage[T1]{fontenc}
\usepackage{amsmath}
\usepackage{hyperref}
\usepackage[numbers,sort&compress]{natbib}
\usepackage{bm}
 \usepackage{algorithm,algorithmic}

\usepackage{mathtools}
\usepackage{amsmath}
\usepackage[shortlabels]{enumitem}

\usepackage{graphicx,epic,eepic,epsfig,amsmath,latexsym,amssymb,verbatim,color}
 
\usepackage{amsfonts}       
\usepackage{nicefrac}       

\usepackage{amsmath}
\usepackage{bbm}

\usepackage{float}
\usepackage{tikz}
\usetikzlibrary{chains}
\usetikzlibrary{fit}
\usepackage{pgflibraryarrows}		
\usepackage{pgflibrarysnakes}		

\usepackage{epsfig}
\usetikzlibrary{shapes.symbols,patterns} 
\usepackage{pgfplots}

\usepackage[strict]{changepage}

\usepackage[marginal]{footmisc}
\usepackage{url}
\usepackage{theorem}
\newtheorem{definition}{Definition}
\newtheorem{proposition}[definition]{Proposition}
\newtheorem{lemma}[definition]{Lemma}

\newtheorem{theorem}[definition]{Theorem}

\def\squareforqed{\hbox{\rlap{$\sqcap$}$\sqcup$}}
\def\qed{\ifmmode\squareforqed\else{\unskip\nobreak\hfil
\penalty50\hskip1em\null\nobreak\hfil\squareforqed
\parfillskip=0pt\finalhyphendemerits=0\endgraf}\fi}
\def\endenv{\ifmmode\;\else{\unskip\nobreak\hfil
\penalty50\hskip1em\null\nobreak\hfil\;
\parfillskip=0pt\finalhyphendemerits=0\endgraf}\fi}
\newenvironment{proof}{\noindent \textbf{{Proof~} }}{\hfill $\blacksquare$}

\newcounter{remark}

\newcounter{example}

\mathchardef\ordinarycolon\mathcode`\:
\mathcode`\:=\string"8000
\def\vcentcolon{\mathrel{\mathop\ordinarycolon}}
\begingroup \catcode`\:=\active
  \lowercase{\endgroup
  \let :\vcentcolon
  }

\usepackage{cleveref}
\usepackage{graphicx}
\usepackage{xcolor}

\RequirePackage[framemethod=default]{mdframed}
\newmdenv[skipabove=7pt,
skipbelow=7pt,
backgroundcolor=darkblue!15,
innerleftmargin=5pt,
innerrightmargin=5pt,
innertopmargin=5pt,
leftmargin=0cm,
rightmargin=0cm,
innerbottommargin=5pt,
linewidth=1pt]{tBox}

\newmdenv[skipabove=7pt,
skipbelow=7pt,
backgroundcolor=blue2!25,
innerleftmargin=5pt,
innerrightmargin=5pt,
innertopmargin=5pt,
leftmargin=0cm,
rightmargin=0cm,
innerbottommargin=5pt,
linewidth=1pt]{dBox}
\newmdenv[skipabove=7pt,
skipbelow=7pt,
backgroundcolor=darkkblue!15,
innerleftmargin=5pt,
innerrightmargin=5pt,
innertopmargin=5pt,
leftmargin=0cm,
rightmargin=0cm,
innerbottommargin=5pt,
linewidth=1pt]{sBox}
\definecolor{darkblue}{RGB}{0,76,156}
\definecolor{darkkblue}{RGB}{0,0,153}
\definecolor{blue2}{RGB}{102,178,255}
\definecolor{darkred}{RGB}{195,0,0}

\newcommand{\nc}{\newcommand}
\nc{\rnc}{\renewcommand}
\nc{\beg}{\begin{equation}}
\nc{\eeq}{{\end{equation}}}
\nc{\beqa}{\begin{eqnarray}}
\nc{\eeqa}{\end{eqnarray}}
\nc{\lbar}[1]{\overline{#1}}
\nc{\bra}[1]{\langle#1|}
\nc{\ket}[1]{|#1\rangle}
\nc{\ketbra}[2]{|#1\rangle\!\langle#2|}
\nc{\braket}[2]{\langle#1|#2\rangle}

\nc{\proj}[1]{| #1\rangle\!\langle #1 |}
\nc{\avg}[1]{\langle#1\rangle}
\nc{\rank}{\operatorname{Rank}}
\nc{\smfrac}[2]{\mbox{$\frac{#1}{#2}$}}
\nc{\tr}{\operatorname{Tr}}
\nc{\ox}{\otimes}
\nc{\dg}{\dagger}
\nc{\dn}{\downarrow}
\nc{\cA}{{\cal A}}
\nc{\cB}{{\cal B}}
\nc{\cC}{{\cal C}}
\nc{\cD}{{\cal D}}
\nc{\cE}{{\cal E}}
\nc{\cF}{{\cal F}}
\nc{\cG}{{\cal G}}
\nc{\cH}{{\cal H}}
\nc{\cI}{{\cal I}}
\nc{\cJ}{{\cal J}}
\nc{\cK}{{\cal K}}
\nc{\cL}{{\cal L}}
\nc{\cM}{{\cal M}}
\nc{\cN}{{\cal N}}
\nc{\cO}{{\cal O}}
\nc{\cP}{{\cal P}}
\nc{\cQ}{{\cal Q}}
\nc{\cR}{{\cal R}}
\nc{\cS}{{\cal S}}
\nc{\cT}{{\cal T}}
\nc{\cV}{{\cal V}}
\nc{\cX}{{\cal X}}
\nc{\cY}{{\cal Y}}
\nc{\cZ}{{\cal Z}}
\nc{\cW}{{\cal W}}
\nc{\csupp}{{\operatorname{csupp}}}
\nc{\qsupp}{{\operatorname{qsupp}}}
\nc{\var}{{\operatorname{var}}}
\nc{\rar}{\rightarrow}
\nc{\lrar}{\longrightarrow}
\nc{\polylog}{{\operatorname{polylog}}}
\nc{\wt}{{\operatorname{wt}}}
\nc{\av}[1]{{\left\langle {#1} \right\rangle}}
\nc{\supp}{{\operatorname{supp}}}

\nc{\argmin}{{\operatorname{argmin}}}

\def\a{\alpha}
\def\b{\beta}

\def\x{\xi}

\nc{\RR}{{{\mathbb R}}}
\nc{\CC}{{{\mathbb C}}}
\nc{\FF}{{{\mathbb F}}}
\nc{\NN}{{{\mathbb N}}}
\nc{\ZZ}{{{\mathbb Z}}}
\nc{\PP}{{{\mathbb P}}}
\nc{\QQ}{{{\mathbb Q}}}
\nc{\UU}{{{\mathbb U}}}
\nc{\EE}{{{\mathbb E}}}
\nc{\id}{{\operatorname{id}}}

\nc{\CHSH}{{\operatorname{CHSH}}}

\nc{\be}{\begin{equation}}
\nc{\ee}{{\end{equation}}}
\nc{\bea}{\begin{eqnarray}}
\nc{\eea}{\end{eqnarray}}
\nc{\<}{\langle}
\rnc{\>}{\rangle}
\nc{\rU}{\mbox{U}}

\nc{\ob}[1]{#1}

\nc{\SEP}{{\text{\rm SEP}}}
\nc{\NS}{{\text{\rm NS}}}
\nc{\LOCC}{{\text{\rm LOCC}}}
\nc{\PPT}{{\text{\rm PPT}}}
\nc{\EXT}{{\text{\rm EXT}}}
\nc{\Sym}{{\operatorname{Sym}}}

\nc{\ERLO}{{E_{\text{r,LO}}}}
\nc{\ERLOCC}{{E_{\text{r,LOCC}}}}
\nc{\ERPPT}{{E_{\text{r,PPT}}}}
\nc{\ERLOCCinfty}{{E^{\infty}_{\text{r,LOCC}}}}
\nc{\Aram}{{\operatorname{\sf A}}}

\usepackage{tikz}
 
\makeatletter
\def\grd@save@target#1{%
  \def\grd@target{#1}}
\def\grd@save@start#1{%
  \def\grd@start{#1}}
\tikzset{
  grid with coordinates/.style={
    to path={%
      \pgfextra{%
        \edef\grd@@target{(\tikztotarget)}%
        \tikz@scan@one@point\grd@save@target\grd@@target\relax
        \edef\grd@@start{(\tikztostart)}%
        \tikz@scan@one@point\grd@save@start\grd@@start\relax
        \draw[minor help lines,magenta] (\tikztostart) grid (\tikztotarget);
        \draw[major help lines] (\tikztostart) grid (\tikztotarget);
        \grd@start
        \pgfmathsetmacro{\grd@xa}{\the\pgf@x/1cm}
        \pgfmathsetmacro{\grd@ya}{\the\pgf@y/1cm}
        \grd@target
        \pgfmathsetmacro{\grd@xb}{\the\pgf@x/1cm}
        \pgfmathsetmacro{\grd@yb}{\the\pgf@y/1cm}
        \pgfmathsetmacro{\grd@xc}{\grd@xa + \pgfkeysvalueof{/tikz/grid with coordinates/major step}}
        \pgfmathsetmacro{\grd@yc}{\grd@ya + \pgfkeysvalueof{/tikz/grid with coordinates/major step}}
        \foreach \x in {\grd@xa,\grd@xc,...,\grd@xb}
        \node[anchor=north] at (\x,\grd@ya) {\pgfmathprintnumber{\x}};
        \foreach \y in {\grd@ya,\grd@yc,...,\grd@yb}
        \node[anchor=east] at (\grd@xa,\y) {\pgfmathprintnumber{\y}};
      }
    }
  },
  minor help lines/.style={
    help lines,
    step=\pgfkeysvalueof{/tikz/grid with coordinates/minor step}
  },
  major help lines/.style={
    help lines,
    line width=\pgfkeysvalueof{/tikz/grid with coordinates/major line width},
    step=\pgfkeysvalueof{/tikz/grid with coordinates/major step}
  },
  grid with coordinates/.cd,
  minor step/.initial=.2,
  major step/.initial=1,
  major line width/.initial=2pt,
}
\makeatother

\usepackage{thmtools}
\usepackage{thm-restate}
\usepackage{etoolbox}
\makeatletter
\def\problem@s{}
\newcounter{problems@cnt}

\newcommand{\allproblems}{\problem@s}
\makeatother

\definecolor{colorone}{rgb}{1,0.36,0.03}
\definecolor{colortwo}{rgb}{0.4,0.77,0.17}
\definecolor{colorthree}{rgb}{0.01,0.51,0.93}
\definecolor{colorfour}{rgb}{0.47,0.26,0.58}
\usepackage{tcolorbox}
\usepackage{relsize}
\usepackage{graphicx}
\nc{\st}{\text{subject to} \ }
\nc{\supre}{\text{supremum} \ }
\nc{\sdp}{\text{sdp}}
\nc{\cU}{\mathcal U}
\usepackage[qm]{qcircuit}
\usepackage{array}

\newcommand{\xw}[1]{\textcolor{black}{#1}}

\allowdisplaybreaks

\begin{document}
\title{Variational Quantum Singular Value Decomposition}
\author{Xin Wang}
\email{wangxin73@baidu.com}
\affiliation{Institute for Quantum Computing, Baidu Research, Beijing 100193, China}
\author{Zhixin Song}
\email{zhixinsong0524@gmail.com}
\affiliation{Institute for Quantum Computing, Baidu Research, Beijing 100193, China}
\author{Youle Wang}
\email{youle.wang@student.uts.edu.au}
\affiliation{Institute for Quantum Computing, Baidu Research, Beijing 100193, China}
\affiliation{Centre for Quantum Software and Information,
University of Technology Sydney, NSW 2007, Australia}

\begin{abstract} 
Singular value decomposition is central to many problems in engineering and scientific fields. Several quantum algorithms have been proposed to determine the singular values and their associated singular vectors of a given matrix. Although these algorithms are promising, the required quantum subroutines and resources are too costly on near-term quantum devices. In this work, we propose a variational quantum algorithm for singular value decomposition (VQSVD). By exploiting the variational principles for singular values and the Ky Fan Theorem, we design a novel loss function such that two quantum neural networks (or parameterized quantum circuits) could be trained to learn the singular vectors and output the corresponding singular values. Furthermore, we conduct numerical simulations of VQSVD for random matrices as well as its applications in image compression of handwritten digits. Finally, we discuss the applications of our algorithm in recommendation systems and polar decomposition. Our work explores new avenues for quantum information processing beyond the conventional protocols that only works for Hermitian data, and reveals the capability of matrix decomposition on near-term quantum devices. 
\end{abstract}  

\maketitle 

\section{Introduction}

Matrix decompositions are integral parts of many algorithms in
optimization~\cite{Boyd2004}, machine learning~\cite{Murphy2012}, and recommendation systems~\cite{Koren2009}. One crucial approach is the singular value decomposition (SVD). Mathematical applications of the SVD include computing the pseudoinverse, matrix approximation, and estimating the range and null space of a matrix. SVD has also 
been successfully applied to many areas of science and engineering industry, such as data compression, noise reduction, and image processing. 
The goal of 
SVD is to decompose a square matrix $M$ to $UDV^{\dagger}$ with diagonal matrix $D = \text{diag}(d_1, \cdots, d_r)$ and unitaries $U$ and $V$, where $r$ denotes the rank of matrix $M$.

Quantum computing is believed to deliver new technology to speed up computation, and it already promises speedups for integer factoring~\cite{Shor1997} and database search~\cite{Grover1996} in theory. Enormous efforts have been made in exploring the possibility of using quantum resources to speed up other important tasks, including linear system solvers~\cite{Harrow2009,Clader2013,Childs2017b,Wossnig2018a,Subasi2018a}, convex optimizations~\cite{Brandao2017,Chakrabarti2020,Brandao2019,VanApeldoorn2020}, and machine learning~\cite{Biamonte2017b,Schuld2018b,Ciliberto2018}. 
Quantum algorithms for SVD have been proposed in \cite{Kerenidis2016,Rebentrost2018}, which leads to applications in solving linear systems of equations~\cite{Wossnig2018a} and developing quantum recommendation systems~\cite{Kerenidis2016}. 
However, these algorithms above are too costly to be convincingly validated for near-term quantum devices, which only support a restricted number of physical qubits and limited gate fidelity.  

Hence, an important direction is to find useful algorithms that could work on noisy intermediate-scale quantum (NISQ) devices~\cite{Preskill2018a}. The leading strategy to solve various problems using NISQ devices are called variational quantum algorithms \cite{McClean2016,Cerezo2020,Endo2020,Bharti2021}. 
These algorithms can be implemented on shallow-depth quantum circuits that depend on external parameters (e.g., angle $\theta$ in rotation gates $R_y(\theta)$), which are also known as parameterized quantum circuits or quantum neural networks (QNNs). These parameters will be optimized externally by a classical computer with respect to certain loss functions. Various variational algorithms using QNNs have been proposed for Hamiltonian ground and excited states preparation~\cite{Peruzzo2014,McClean2017a,Higgott2018,Barkoutsos2019,Nakanishi2019,Ostaszewski2019,nakanishi2019subspace}, quantum state metric estimation~\cite{Cerezo2019,Chen2020}, Gibbs state preparation~\cite{Wu2019b,Chowdhury2020,Wang2020a}, quantum compiling~\cite{Khatri2019,Heya2018,Jones2018,Sharma_2020}, machine learning \cite{Schuld2018b,Ciliberto2018,Cao2021,Cong2018,Li2021} etc. 
Furthermore, unlike the strong need of error correction in fault-tolerant quantum computation, noise in shallow quantum circuits can be suppressed via error mitigation \cite{Temme2017,McArdle2018,Strikis2020,Kandala2018,Jiang2020,Wang2021}, indicating the feasibility of quantum computing with NISQ devices.

In this paper, we formulate the task of SVD as an optimization problem and derive a variational quantum algorithm for singular value decomposition (VQSVD) that can be implemented on near-term quantum computers. The core idea is to construct a novel loss function inspired by the variational principles and properties of singular values. We theoretically show that the optimized quantum neural networks based on this loss function could learn the singular vectors of a given matrix. That is, we could train two quantum neural networks $U(\a)$ and $V(\b)$ to learn the singular vectors of a matrix $M$ in the sense that $M \approx U(\a) D V(\b)^{\dagger}$, where the diagonal matrix $D$ provides us the singular values.
Our approach generalizes the conventional methods of Hamiltonian diagonalization \cite{Larose2019,Nakanishi2019} to a non-Hermitian regime, extending the capabilities of matrix decomposition on near-term quantum computers. As a proof of principle, we conduct numerical simulations to estimate the SVD of random $8\times 8$ matrices. Furthermore, we explore the possibility of applying VQSVD to compress images of size $32 \times 32$ pixel, including the famous MNIST dataset. Finally, we showcase the applications of VQSVD in recommendation systems and polar decomposition.




\section{Main results}
\subsection{Variational quantum  singular value decomposition} \label{sec:VQSVD}
In this section, we present a variational quantum algorithm for singular value decomposition of  $n\times n$  matrices, and it can be naturally generalized for $n \times m$ complex matrices.
For given $n\times n$ matrix $M \in \RR^{n\times n}$, there exists a decomposition of the form
$M = U D V^{\dagger}$,
where $U,V \in \RR^{n\times n}$ are unitary operators and $D$ is a diagonal matrix with $r$ positive entries $d_{1},\cdots,d_r$ and $r$ is the rank of $M$. Alternatively, we write 
$M = \sum_{j=1}^r d_j \ketbra{u_j}{v_j}$, where $\ket{u_j}$, $\ket{v_j}$, and $d_j$ are the sets of left and right 
orthonormal singular vectors, and singular values of $M$, respectively.

A vital issue in NISQ algorithm is to choose a suitable loss function. A desirable loss function here should be able to output the target singular values and vectors after the optimization and in particular should be implementable on near-term devices. As a key step, we design such a desirable loss function for quantum singular value decomposition (cf.~Section~\ref{sec:cost}). 

The input of our VQSVD algorithm is a decomposition of the matrix $M$ into a linear combination of $K$ unitaries of the form
$M = \sum_{k=1}^K c_k A_k$
with real numbers $c_k$. For instance, one could decompose $M$ into a linear combination of Pauli terms. 
\xw{And a method for finding such a decomposition was proposed in a recent work~\cite{Gunlycke2020}. Several important generic matrices in engineering and science also exhibit LCUs, including Toeplitz, circulant, and Hankle matrices \cite{wan2019implementing}.} After taking in the inputs, our VQSVD algorithm enters a hybrid quantum-classical optimization loop to train the parameters $\bm\a$ and $\bm\b$ in the parameterized quantum circuits $U(\bm\a)$ and $V(\bm\b)$ via a designed loss $L(\bm\a,\bm\b)$ (cf.~Section~\ref{sec:cost}). This loss function can be computed on quantum computers via the Hadamard tests. We then feeds the value of the loss function or its gradients (in gradient-based optimization) to a classical computer, which adjusts the parameters $\bm\a$ and $\bm\b$ for the next round of the loop. The goal is to find the global minimum of $L(\bm\a,\bm\b)$, i.e.,
$\bm\a^*,\bm\b^* \coloneqq \arg\min_{\bm\a,\bm\b} L(\bm\a,\bm\b)$.
 
In practice, one will need to set some termination condition (e.g., convergence tolerance) on the optimization loop. After the hybrid optimization, one will obtain values $\{m_{j}\}_{j=1}^{T}$ and optimal parameters $\bm\a^{*}$ and $\bm\b^{*}$. The outputs $\{m_{j}\}_{j=1}^{T}$ approximate the singular values of $M$, and approximate singular vectors of $M$ are obtained by inputting optimal parameters $\bm\a^{*}$ and $\bm\b^{*}$ into the parameterized circuits $U$ and $V$ in VQSVD and then applying to the orthonormal vectors $\ket{\psi_j}$ for all $j=1,...,T$. The detailed VQSVD algorithm is included in Algorithm \ref{alg:rftl}. A schematic diagram is shown in Fig.~\ref{VQSVD_Diagram}. 

\begin{algorithm}[H] 
\caption{Variational quantum  singular value decomposition (VQSVD)}
\begin{algorithmic}[1] \label{alg:rftl}
\STATE Input:  $\{c_k, A_k\}_{k=1}^K$, desired rank $T$, parametrized circuits $U(\bm\a)$ and $V(\bm\b)$ with initial parameters of $\bm\a$,  $\bm\b$, and tolerance $\varepsilon$;

\STATE Prepare positive numbers $q_1 > \cdots q_{T} >0$;

\STATE Choose computational basis $\ket{\psi_1},\cdots,\ket{\psi_T}$;

\FOR {$j=1,\cdots,T$}
\STATE  \qquad Apply $U(\bm\a)$ to state $		\ket{\psi_j}$ and obtain $\ket {u_j} = U(\bm\a)\ket{\psi_j}$; 
 
\STATE \qquad Apply $V(\bm\b)$ to state $\ket{\psi_j}$ and obtain $\ket {v_j} = V(\bm\b)\ket{\psi_j}$ ;

\STATE \qquad Compute $m_j=\text{Re}\bra{u_j} M \ket{v_j}$ via Hadamard tests;
\ENDFOR 

\STATE Compute the loss function $L(\bm\a,\bm\b) =  \sum_{j=1}^T q_j m_j$;

\STATE Perform optimization to maximize $L(\bm\a,\bm\b)$, update parameters of $\bm\a$ and $\bm\b$;

\STATE Repeat 4-10 until the loss function $L(\bm\a,\bm\b)$ converges with tolerance $\varepsilon$;
\STATE Output $\{m_j\}_{j=1}^T$  as the largest $T$ singular values, output $U(\bm\a^*)$ and $V(\bm\b^*)$ as corresponding unitary operators ($\bra{\psi_j}U(\bm\a^*)^\dagger$ and $V(\bm\b^*)\ket{\psi_j}$ are left and right singular vectors, respectively).
\end{algorithmic}
\end{algorithm}

\begin{figure*}[t]
\centering
\includegraphics[scale = 0.47]{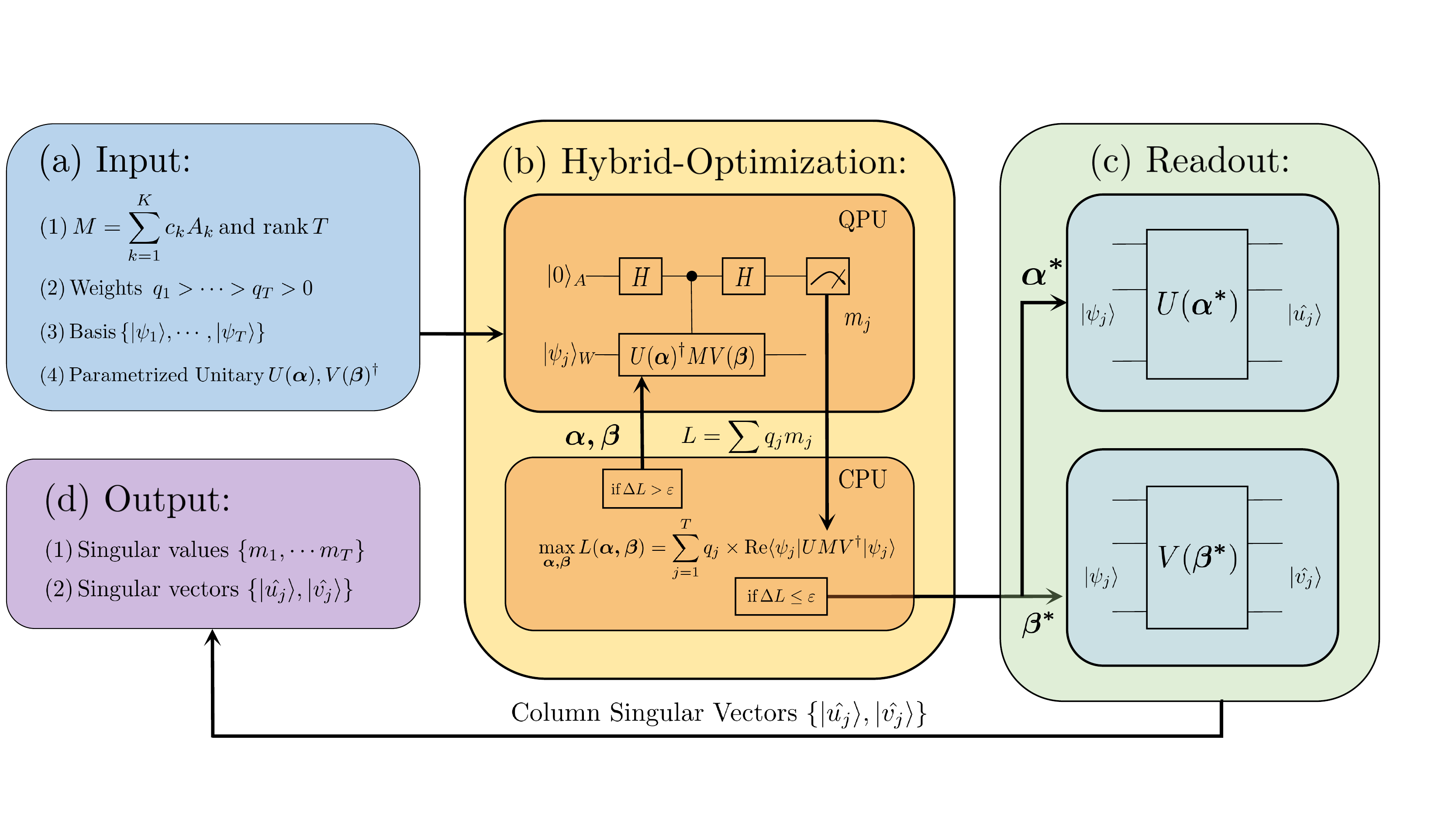}
\caption{Schematic diagram including the steps of VQSVD algorithm. (a) The unitary decomposition of matrix $M$ is provided as the first input. This can be achieved through Pauli decomposition with tensor products of $\{X, Y, Z, I\}$. The algorithm also requires the desired number of singular values $T$, orthonormal input states $\langle\psi_i |\psi_j\rangle = \delta_{ij}$ and two parametrized unitary matrices $U(\bm\a), V(\bm \b)$. Finally, the weight is usually set to be integers $q_j = T+1-j$. The former two information will be sent to the hybrid-optimization loop in (b) where the quantum computer (QPU) will estimate each singular value $m_j = \text{Re} \bra{\psi_j}U^\dagger MV \ket{\psi_j}$ via Hadamard test. These estimations are sent to a classical computer (CPU) to evaluate the loss function until it converges to tolerance $\varepsilon$. Once we reach the global minimum, the singular vectors $\{|\hat{u_j}\rangle, |\hat{v_j}\rangle\}$ can be produced in (c) by applying the learned unitary matrices $U(\bm \a^*)$ and $V(\bm\b^*)$ on orthonormal basis $\{\ket{\psi_j}\}_{j=1}^{T}$ to extract the column vectors.}
\label{VQSVD_Diagram} 
\end{figure*}

\subsection{Loss function}\label{sec:cost}
In this section, we provide more details and intuitions of the loss function in VQSVD. The key idea is to exploit the variational principles in matrix computation, which have great importance in analysis for error bounds of matrix analysis. In particular, the singular values satisfy a subtler variational property that incorporates both left and right singular vectors at the same time. For a given $n\times n$ matrix $M$, the largest singular value of $M$ can be characterized by
\begin{align}
d_1 = \max_{\ket u,\ket v} \frac{|\bra u M \ket v|}{\|u\|\|v\|} = \max_{\ket u,\ket v\in\cS} {\text{Re}[\bra u M \ket v}],
\end{align}
where $\cS$ is the set of pure states (normalized vectors) and $\text{Re}$ means taking the real part. Moreover, by denoting the optimal singular vectors as $\ket{u_1},\ket{v_1}$, the remaining singular values ($d_2\ge \cdots \ge d_r$) can be deduced using similar methods by restricting the unit vectors to be orthogonal to previous singular value vectors.

For a given $n\times n$ matrix $M$, the largest singular value of $M$ can be characterized by
\begin{align}\label{eq:SV1}
d_1 = \max_{\ket u,\ket v\in\cS} {\text{Re}[\bra u M \ket v}],
\end{align}
where $\cS$ is the set of pure states (normalized vectors) and $\text{Re}[\cdot]$ means to take the real part. Moreover, by denoting the optimal singular vectors as $\ket{u_1},\ket{v_1}$, the remaining singular values ($d_2\ge \cdots \ge d_r$) can be deduced as follows
\begin{equation}\label{eq:SV2}
\begin{split}
d_k = \max & \ \text{Re}[\bra u M \ket v]\\
\text{s.t.} &\ \ket u,\ket v\in\cS, \\
&\ \ket{u} \bot\rm{span}\{\ket{u_1},\cdots,\ket{u_{k-1}}\},\\
&\ \ket{v} \bot\rm{span}\{\ket{v_1},\cdots,\ket{v_{k-1}}\}.\\
\end{split}
\end{equation}

Another useful fact~(Ky Fan Theorem, cf.~\cite{fan1951maximum,Zhang2015a}) is that 
\begin{align}\label{eq:property VSVD}
\sum_{j=1}^{T} d_j = \max_{\text{orthonomal} \ \{u_j\},\{v_j\} } \sum_{j=1}^{T} \bra {u_j} M \ket {v_j}.
\end{align} 
 
For a given matrix $M$,
the loss function in our VQSVD algorithm is defined as
\begin{align}\label{eq:lss}
L(\bm\a,\bm\b) = \sum_{j=1}^T q_j\times \text{Re} \bra{\psi_j}U(\bm\a)^{\dagger} M V(\bm\b)\ket{\psi_j},
\end{align}
 where $q_1>\cdots>q_T>0$ are real weights and $\{{\psi_j}\}_{j=1}^T$ is a set of orthonormal states. The setting of constants $q_j$ not only theoretically guarantees the correct outcome but also allows our approach to be more flexible. Usually, choosing these constants with excellent performances are empirical. 


\begin{theorem}
For a given matrix $M$, the loss function $L(\bm\a,\bm\b)$  is maximized if only if \begin{align}
\bra{\psi_j}U(\bm\a)^{\dagger} M V(\bm\b)\ket{\psi_j} = d_j, \ \forall \ 1\le j\le T,
\end{align}
where $d_1\ge \cdots \ge d_T$  are the largest $T$ singular values  of $M$ and $\ket{\psi_1},\cdots,\ket{\psi_T}$ are orthonormal vectors. Moreover, $\bra{\psi_j}U(\bm\a)^\dagger$ and $V(\bm\b)\ket{\psi_j}$ are left and right singular vectors, respectively.
\end{theorem}

\begin{proof}
Let assume that $\bra{\psi_j}U(\bm\a)^{\dagger} M V(\bm\b)\ket{\psi_j} = m_j$ are real numbers for simplicity, which could be achieved after the ideal maximization process.

Then we have
\begin{align}
L(\bm\a,\bm\b) = &\sum_{j=1}^T q_j\times m_j\\
 \le & \sum_{j=1}^T q_j\times m_j^\downarrow \label{eq:paixu}\\
= & \sum_{j=1}^{T}(q_j-q_{j+1})\sum_{t=1}^{j}m_t^\downarrow \\
\le &\sum_{j=1}^{T}(q_j-q_{j+1})\sum_{t=1}^{j} d_t.\label{eq:loss upper bound}
\end{align}

Assume $q_{T+1}=0$. The first inequality Eq.~\eqref{eq:paixu} follows due to the rearrangement inequality.
The second inequality Eq.~\eqref{eq:loss upper bound} follows due to property of singular values in Eq.~\eqref{eq:property VSVD}. 
Note that the upper bound in Eq.~\eqref{eq:loss upper bound} could be achieved if and only if $\sum_{t=1}^jd_t = \sum_{t=1}^j m_t$ for all $j$, which is equivalent to
\begin{align}
m_j = d_j, \quad\forall \ 1\le j \le T.
\end{align} 

Therefore, the loss function is maximized if and only if $\bra{\psi_j}U(\bm\a)^{\dagger} M V(\bm\b)\ket{\psi_j}$ extracts the singular values $d_j$ for each $j$ from $1$ to $T$. Further, due to the variational property of singular values in Eq.~\eqref{eq:SV1} and Eq.~\eqref{eq:SV2}, we conclude that the quantum neural networks $U$ and $V$ learn the singular vectors in the sense that $\bra{\psi_j}U(\bm\a)^\dagger$ and $V(\bm\b)\ket{\psi_j}$ extract the left and right singular vectors, respectively.
\end{proof}

\xw{
In particular, the weights $q_{j}$ can be any sequence that satisfies the condition in Theorem 1. They are the key to train the quantum neural networks to learn the right singular vectors. Such weights can be understood as the tool to discriminate and identify the eigen-space. }

\begin{proposition}
\label{prop:loss:function:estimation}
The loss function $L(\bm\a,\bm\b)$ can be estimated on near-term quantum devices.
\end{proposition}

For $M = \sum_{k=1}^K c_k A_k$ with unitaries $A_k$ and real numbers $c_k$, the quantity $\text{Re}\bra{\psi_j}U(\bm\a)^{\dagger} M V(\bm\b)\ket{\psi_j}$ can be decomposed to 
\begin{align}\label{eq:decom of M value}
    \sum_{k=1}^K  c_k \times \text{Re}\bra{\psi_j}U(\bm\a)^{\dagger} A_k V(\bm\b)\ket{\psi_j}.
\end{align}
To estimate the quantity in Eq.~\eqref{eq:decom of M value}, we could use quantum subroutines for estimating the quantity ${\rm Re}{\bra \psi U \ket\psi}$ for a general unitary $U$. One of these subroutines is to utilize the well-known Hadamard test~\cite{Aharonov2009a}, which requires only one ancillary qubit, one copy of state $\ket{\psi}$, and one controlled unitary operation $U$, and hence it can be experimentally implemented on near term quantum hardware. To be specific, Hadamard test (see Fig.~\ref{fig:hadamard:test}) starts with state $\ket{+}_{A}\ket{\psi}_{W}$, where $A$ denotes the ancillary qubit, and $W$ denotes the work register, and then apply a controlled unitary $U$, conditioned on the qubit in register $A$, to prepare state $\frac{1}{\sqrt{2}}(\ket{0}_{A}\ket{\psi}_{W}+\ket{1}_{A}U\ket{\psi}_{W})$, at last, apply a Hadamard gate on the ancillary qubit, and measure. If the measurement outcome is $0$, then let the output be $1$; otherwise, let the output be $-1$, and the expectation of output is ${\rm Re}\bra \psi U \ket \psi$. As for the imaginary part ${\rm Im}\bra \psi U \ket \psi$, it also can be estimated via Hadamard test by starting with state $\frac{1}{\sqrt{2}}(\ket{0}_{A}+i\ket{1}_{A})\ket{\psi}_{W}$. 

\xw{\textbf{Remark 1}: Notice that terms in Eq.~\eqref{eq:decom of M value} may be exponentially many for matrices with particular real-world applications, endowing a potential obstacle to the loss evaluation. However, we could apply the importance sampling technique to circumvent this issue. We provide a detailed discussion in Appendix~\ref{sec:supplemental:cost_evaluation}. In particular, we show that the loss evaluation efficiency is dependent on the $\ell_1$-norm of all coefficients $c_k$ rather than the integer $K$. Hence, for a matrix with a reasonable $\ell_1$-norm (e.g., polynomial size), the evaluation could be efficient even for large-size problems. Besides, we take the circulant matrix as an example in Appendix~\ref{sec:supplemental:cost_evaluation} for illustration, suggesting the potential application of our approach in related practical fields.} 

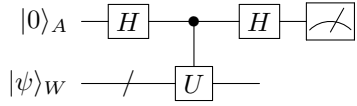
\begin{figure}[t]
\[\Qcircuit @C=1em @R=1em{
\lstick{\ket{0}_{A}}  & \gate{H} & \ctrl{1} &  \gate{H} & \meter\\ 
\lstick{\ket{\psi}_{W}} & /\qw   &\gate{U}  & \qw &  \\
}
\]
\caption{Quantum circuit for implementing Hadamard test}
\label{fig:hadamard:test}
\end{figure}

\subsection{As a generalization of the VQE family}\label{sec:vqe}
In this subsection, we discuss the connection between our VQSVD algorithm and the variational quantum eigensolver (VQE)~\cite{Peruzzo2014}, which estimates the ground-state energy given a molecular electronic-structure Hamiltonian $H$. This is a central problem to quantum chemistry as many properties of the molecule can be calculated once we determine the eigenstates of $H$. Several related algorithms have been proposed to improve the spectrum learning capability of VQE (i.e., SSVQE~\cite{nakanishi2019subspace}) such that the properties of excited states can also be explored. We consider these algorithms as a unified family and promising applications for quantum chemistry. 

For practical reasons, one has to discretize the aforementioned Hamiltonian $H$ into a series of Pauli tensor products $d_i \otimes \cdots \otimes d_j$ to work on quantum devices. Many physical models naturally fit into this scheme including the quantum Ising model and the Heisenberg Model. In particular, if the input of VQSVD in Eq.~\eqref{eq:lss} is a Hamiltonian in a discretized form (i.e., a linear combination of Pauli tensor products or equivalently a Hermitian matrix), VQSVD 
could be naturally applied to diagonalize this Hamiltonian and prepare the eigenstates.  
Therefore, VQSVD can be seen as a generalization of the VQE family that works not only for Hamiltonians, but also for more general 
non-Hermitian matrices.

\section{Optimization of the loss function}
Finding optimal parameters $\{\bm\a^*, \bm\b^*\}$ is a significant part of variational hybrid quantum-classical algorithms. Both gradient-based and gradient-free methods could be used to do the optimization. Here, we provide analytical details on the gradient-based approach, and we refer to \cite{Benedetti2019a} for more information on the optimization subroutines in variational quantum algorithms.
Reference about gradients estimation via quantum devices can be found in Ref. \cite{Mitarai2018,Schuld2019,Ostaszewski2019}.

\subsection{Gradients estimation}
Here, we discuss the computation of the gradient of the global loss function $L(\bm\alpha,\bm\beta)$ by giving an analytical expression, and show that the gradients can be estimated by shifting parameters of the circuit used for evaluating the loss function. In Algorithm~\ref{alg:rftl}, to prepare states $\ket{u_{j}}$ and $\ket{v_{j}}$, we apply gate sequences $U=U_{\ell_1}...U_{1}$ and $V=V_{\ell_2}...V_{1}$ in turn to state $\ket{\psi_{j}}$, where each gate $U_{l}$ and $V_{k}$ are either fixed, e.g., C-NOT gate, or parameterized, for all $l=1,...,\ell_1$ and $k=1,...,\ell_2$. The parameterized gates $U_{l}$ and $V_{k}$ have forms $U_{l}=e^{-iH_{l}\alpha_{l}/2}$ and $V_{k}=e^{-iQ_{k}\beta_{k}/2}$, respectively, where $\alpha_{l}$ and $\beta_{k}$ are real parameters, and $H_{l}$ and $Q_{k}$ are tensor products of Pauli matrices. Hence the gradient of loss function $L$ is dependent on parameters $(\bm\alpha,\bm\beta)$ and the following proposition shows it can be computed on near-term quantum devices. 
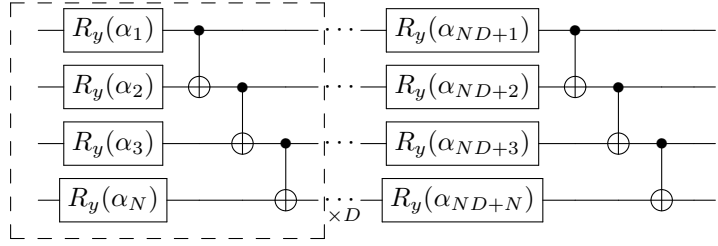
\begin{figure}[t]
\[\Qcircuit @C=0.8em @R=0.5em{
&\gate{R_y(\a_1)}&\ctrl{1}&\qw&\qw&\qw&\cdots& &\gate{R_y(\a_{ND+1})}&\ctrl{1}&\qw&\qw&\qw&\qw \\ 
&\gate{R_y(\alpha_2)}&\targ &\ctrl{1}&\qw&\qw&\cdots& &\gate{R_y(\alpha_{ND+2})}&\targ&\ctrl{1}&\qw&\qw  &\qw\\
&\gate{R_y(\alpha_3)}&\qw&\targ&\ctrl{1}&\qw&\cdots& &\gate{R_y(\alpha_{ND+3})}&\qw&\targ&\ctrl{1}&\qw  &\qw \\
&\gate{R_y(\alpha_N)}&\qw&\qw&\targ&\qw^{\quad \quad \quad \,\times D}&\cdots& &\gate{R_y(\alpha_{ND+N})}&\qw&\qw&\targ&\qw&\qw
\gategroup{1}{1}{4}{5}{2.07em}{--}
}\]
\caption{Hardware-efficient ansatz used in the simulation for both $U(\bm \a)$ and $V(\bm \b)$. The parameters are optimized to minimize the loss function $L(\bm\a, \bm \b)$. $D$ denotes the number of repetitions of the same block (denoted in the dashed-line box) consists of a column of single-qubit rotations about the $y$-axis $R_y(\a_j)$ following by a layer of CNOT gates which only connects the adjacent qubits.}
\label{fig:Ansatz}
\end{figure}

\begin{proposition}
\label{prop:loss:function:gradient}
The gradient of loss function $L(\bm\alpha,\bm\beta)$ can be estimated on near-term quantum devices and its explicitly form is defined as follows,
\begin{align}
\nabla L(\bm\alpha,\bm\beta)=\left(\frac{\partial L}{\partial\alpha_{1}},...,\frac{\partial L}{\partial\alpha_{\ell_1}},\frac{\partial L}{\partial\beta_{1}},...,\frac{\partial L}{\partial\beta_{\ell_2}}\right).
\end{align}
Particularly, the derivatives of $L$ with respect to $\alpha_{l}$ and $\beta_{k}$ can be computed using following formulas, respectively,
\begin{align}
    \frac{\partial L}{\partial \alpha_{l}}&=\frac{1}{2}L(\bm\alpha_{l},\bm\beta),\label{eq:partial_derivative:alpha}\\
    \frac{\partial L}{\partial\beta_{k}}&=\frac{1}{2}L(\bm\alpha,\bm\beta_{k}),\label{eq:partial_derivative_beta}
\end{align}
where notations $\bm\alpha_{l}$ and $\bm\beta_{k}$ denote parameters $\bm\alpha_{l}=(\alpha_{1},...,\alpha_{l}+\pi,...,\alpha_{\ell_1})$ and $\bm\beta_{k}=(\beta_{1},...,\beta_{k}-\pi,...,\beta_{\ell_2})$.
\end{proposition}
\begin{proof}
Notice that the partial derivatives of loss function are given by Eqs.~\eqref{eq:partial_derivative:alpha}~\eqref{eq:partial_derivative_beta}, and hence the gradient is computed by shifting the parameters of circuits that are used to evaluate the loss function. Since the loss function can be estimated on near term quantum devices, claimed in Proposition~\ref{prop:loss:function:estimation}, thus, the gradient can be calculated on near-term quantum devices. 

The derivations of Eqs.~\eqref{eq:partial_derivative:alpha}~\eqref{eq:partial_derivative_beta} use the fact that $\partial_{\theta}({\rm Re}z(\bm\theta))={\rm Re}(\partial_{\theta}z(\bm\theta))$, where $z(\bm\theta)$ is a parameterized complex number, and more details of derivation are deferred to Appendix~\ref{appendix:loss:function:gradient}.
\end{proof}

\subsection{Barren plateaus}
It has been shown that when employing hardware-efficient ansatzes (usually problem agnostic) \cite{Kandala2017},  global cost functions  $\langle\hat{O}\rangle = \tr [\hat{O}U(\bm \theta) \rho U(\bm \theta)^\dagger]$ of an observable $\hat{O}$ are untrainable for large problem sizes since they exhibit exponentially vanishing gradients with respect to the qubit number $N$ which makes the optimization landscape flat (i.e., barren plateaus \cite{McClean2018}). Consequently, traditional optimization methods including the Adam optimizer utilized in our numerical experiments are impacted. This trainability issue happens even when the ansatz is short depth \cite{Cerezoa} and the work~\cite{Sharma2020} showed that the barren plateau phenomenon could also arise in the architecture of dissipative quantum neural networks \cite{Poland2020,Bondarenko2020,Beer2020}.

Several approaches have been proposed to mitigate this problem. One can either implementing identity-block initialization strategy \cite{grant2019initialization} or employing the technique of local cost \cite{Cerezoa}, where the local cost function is defined such that one firstly construct a local observable $\hat{O}_L$, calculating the expectation value with respect to each individual qubit rather than gathering information in a global sense, and finally adding up all the local contributions. The latter strategy has been verified to extend the trainable circuit depth up to $D \in \mathcal{O}(poly(log(N)))$. Recently, a new method of constructing adaptive Hamiltonians during the optimization loop has been proposed to resolve similar issues \cite{Cerezo2020a}. 

\section{Numerical experiments and applications}\label{sec:result}
Here we numerically simulate the VQSVD algorithm with randomly generated $8 \times 8$ non-Hermitian real matrices as a proof of concept. Then, to demonstrate the possibility of scaling VQSVD to larger and more exciting applications, we simulate VQSVD to compress some standard gray-scale $32\times 32$ images. Fig.~\ref{fig:Ansatz} shows the variational ansatz used. Since only real matrices are involved, the combination of $R_y$ rotation gates and CNOT is sufficient. For the case of complex matrices and different ansatzs, see Appendix~\ref{sec:supplemental:numerics}.  We choose the input states to be $\{\ket{\psi_j} \}= \{\ket{000}, \ket{001}, \cdots ,\ket{111}\}$ and the circuit depth $D$ is set to be $D=20$. The parameters $\{\a,\b\}$ are initialized randomly from an uniform distribution $[0, 2\pi]$. All simulations and optimization loop are implemented via Paddle Quantum~\cite{Paddlequantum} on the PaddlePaddle Deep Learning Platform~\cite{Paddle,Ma2019}.

 \begin{figure}[t]
 \centering
\includegraphics[height=65mm,width=0.65\columnwidth]{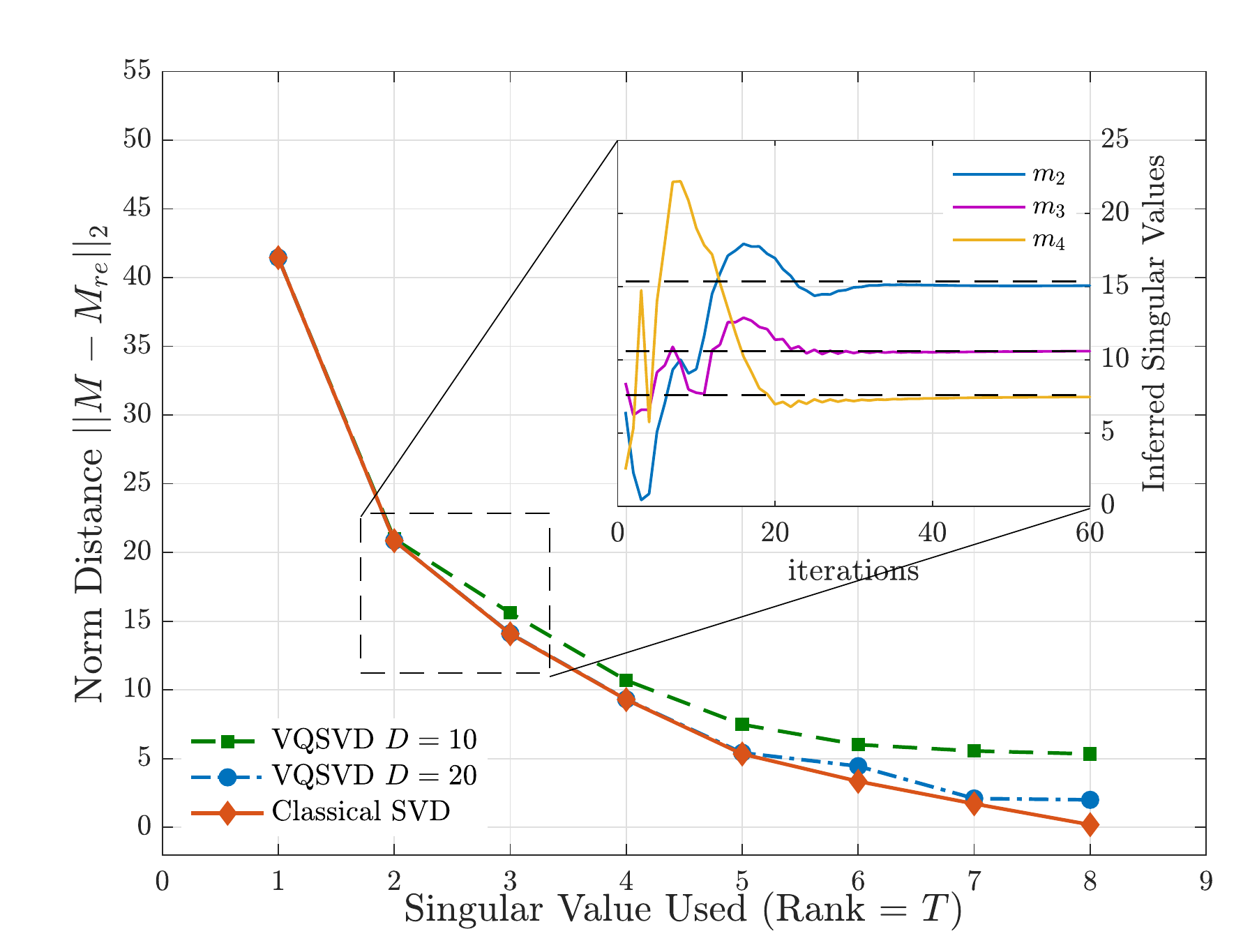}
\caption{\label{fig:VQSVD_performance} Distance measure between the reconstructed matrix $M_{re}$ and the original matrix $M$ via VQSVD with different circuit depth $D = \{ 10, 20\}$ and compare with the classical SVD method. Specifically, we plot the learning process (zoom in the plot) of selected singular values $\{ m_2, m_3, m_4\}$ when using the layered hardware-efficient ansatz illustrated in Fig.~\ref{fig:Ansatz} with circuit depth $D = 20$.}
\end{figure}

\subsection{Three-Qubit example}
The VQSVD algorithm described in Section \ref{sec:VQSVD} can find $T$ largest singular values of matrix $M_{n\times n}$ at once. Here, we choose the weight to be positive integers $(q_1, \cdots, q_T) = (T, T-1,\cdots, 1)$. Fig.~\ref{fig:VQSVD_performance} shows the learning process. One can see that this approach successfully find the desired singular values and one can approximate the original matrix $M$ with inferred singular values and vectors $M_{re}^{(T)} = \sum_{j=1}^{T} m_j\ket{\hat{u_j}}\bra{\hat{v_j}}$. The distance between $M_{re}$ and the original matrix $M$ is taken to be the matrix norm $||A_{n \times n}||_2 = \sqrt{\sum_{i,j=1}^{n} |a_{ij}|^2}$ where $a_{ij}$ are the matrix elements. As illustrated in Fig.~\ref{fig:VQSVD_performance}, the distance decreases as more and more singular values being used.


\subsection{Image compression}
Next, we apply the VQSVD algorithm to compress a $32 \times 32$ pixel handwritten digit image taken from the MNIST dataset with only $7.81\%$ (choose rank to be $T=5$) of its original information. By comparing with the classical SVD method, one can see that the digit $\#7$ is successfully reconstructed with some background noise. Notice the circuit structure demonstrated in Fig.~\ref{fig:Ansatz} is ordinary and it is a not well-studied topic for circuit architecture.  Future studies are needed to efficiently load classical information into NISQ devices beyond the LCU assumption. 
\begin{figure}[h]
\centering
\includegraphics[height=85mm,width=0.65\columnwidth]{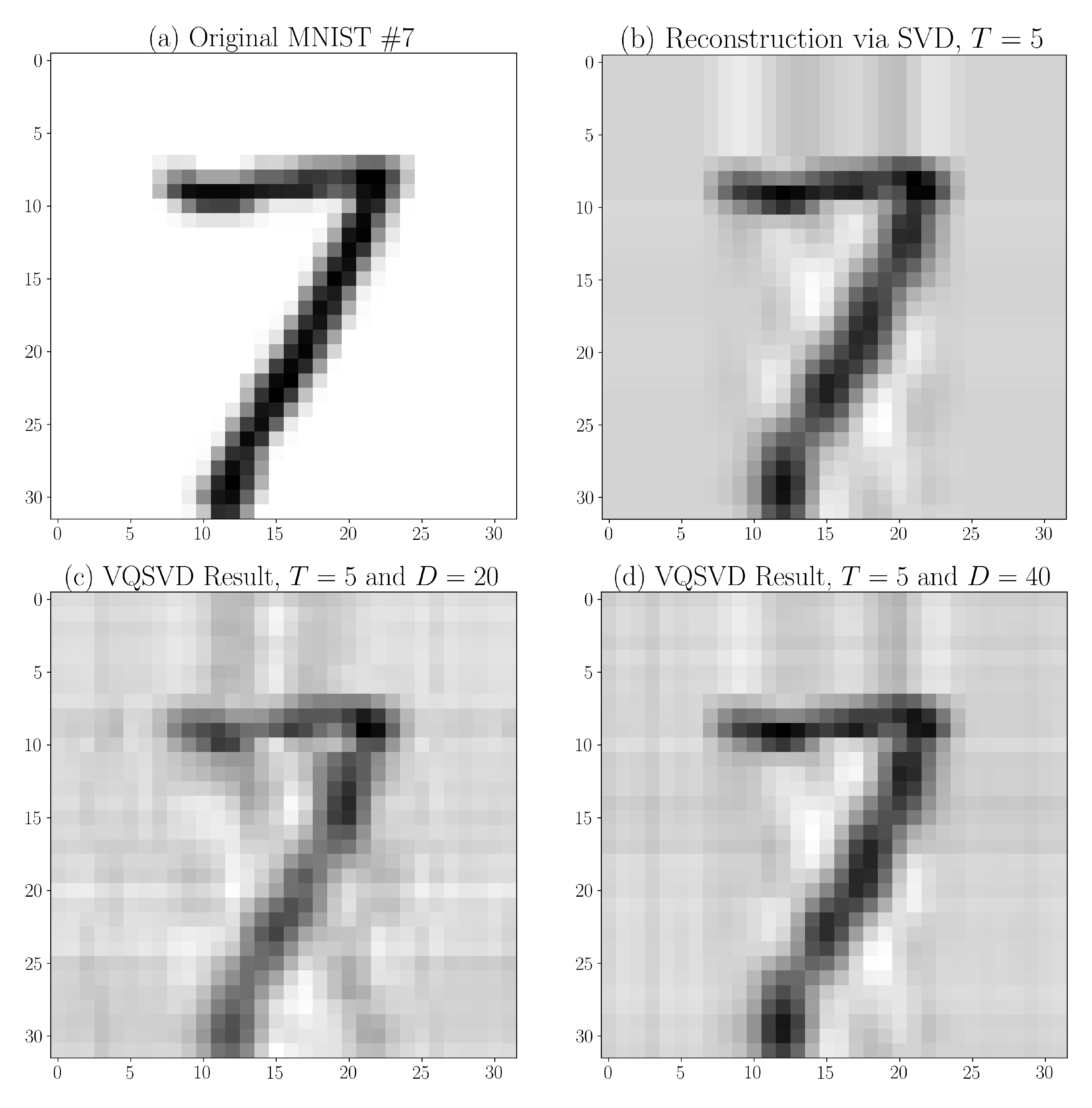}
\caption{\label{fig:image64} Performance of simulated 5-qubit VQSVD for image compression. (a) shows the original handwritten digit $\#7$ and it is compressed via classical SVD up to the largest $5$ singular values in (b). The performance of VQSVD is presented in (c) and (d) with the same rank $T=5$ but different circuit depth $D = \{20, 40\}$. 
}
\end{figure}


\section{Solution quality estimation}\label{sec:verification}
After obtaining the results (singular values and vectors) from Algorithm~\ref{alg:rftl}, it is natural and desirable to have a method to benchmark or verify these results. In this section, we further introduce a procedure for verification purposes. Particularly, we propose a variational quantum algorithm for estimating the sum of largest $T$ squared singular values, i.e., $\sum_{j=1}^T d_{j}^2$, of a matrix $M\in\mathbb{C}^{m\times n}$ as a subroutine. In the following, we first define the error of the inferred singular values and singular vectors, and then show its estimation via a tool provided in Algorithm~\ref{algorithm:vqfne}. 

Let $\{m_{j}\}_{j=1}^{T}$ denote the inferred singular values of the matrix $M$ that are arranged in descending order, and let $\{\ket{\hat{u}_{j}}\}_{j=1}^{T}$ \big($\{\ket{\hat{v}_{j}}\}_{j=1}^{T}$\big) denote the associated inferred left (right) singular vectors. The error of inferred singular values is defined as follows,
\begin{align}
\epsilon_{d}=\sum_{j=1}^{T}(d_{j}-m_{j})^{2},
\label{error:singular:values}
\end{align}
where $d_{j}$s are the exact singular values of matrix $M$ and also arranged in descending order. And the error $\epsilon_{v}$ of inferred singular vectors is defined below,
\begin{align}
\epsilon_{v}=\sum_{j=1}^{T} \parallel H\ket{\hat{e}_{j}^{+}}-m_{j}\ket{\hat{e}_{j}^{+}}\parallel^{2}+\sum_{j=1}^{T} \parallel H\ket{\hat{e}_{j}^{-}}+m_{j}\ket{\hat{e}_{j}^{-}}\parallel^{2},
\label{error:singular:vectors}
\end{align}
where $H$ is a Hermitian operator of the form $H=\op{0}{1}\otimes M+\op{1}{0}\otimes M^{\dagger}$, and $\ket{\hat{e}_{j}^{\pm}}=(\ket{0}\ket{\hat{u}_{j}}\pm\ket{1}\ket{\hat{v}_{j}})/\sqrt{2}$. 
It is worth pointing out that when inferred vectors $\ket{\hat{e}_{j}^{\pm}}$ approximate the eigenvectors $\ket{e_{j}^{\pm}}$ of $H$, i.e., $\epsilon_{v}\to0$, inferred singular vectors $\ket{\hat{u}_{j}}$ and $\ket{\hat{v}_{j}}$ approximate the singular vectors $\ket{u_{j}}$ and $\ket{v_{j}}$ respectively, and vice versa. 

Now we are ready to introduce the procedure for verifying the quality of VQSVD outputs. To quantify the errors, we exploit the fact that errors $\epsilon_{d}$ and $\epsilon_{v}$ are upper bounded. Specifically speaking,
\begin{align}
\epsilon_{d}&\leq\sum_{j=1}^{T}d_{j}^{2}-\sum_{j=1}^{T}m_{j}^{2},
\label{eq:error:d:upper:bound}\\
\epsilon_{v}&\leq 2(\sum_{j=1}^{T}d_{j}^{2}-\sum_{j=1}^{T}m_{j}^{2}),\label{eq:error:v:upper:bound}
\end{align}
We refer the detailed proofs for Eqs.~\eqref{eq:error:d:upper:bound},~\eqref{eq:error:v:upper:bound} to Lemma~\ref{le:error:upper:bound} included in Appendix~\ref{sec:supplemental:verification}.

Thus, we only need to evaluate the sum $\sum_{j=1}^{T}d_{j}^{2}$, which is the Frobenius norm of matrix $M$, and the sum $\sum_{j=1}^{T}m_{j}^{2}$ independently. The latter can be directly computed from the outputs of VQSVD while the former quantity can be estimated through Algorithm~\ref{algorithm:vqfne} presented in Appendix~\ref{sec:supplemental:verification}.

\section{Applications}
In this section, we give a brief discussion about the application of VQSVD in recommendation systems, and some other problems such as polar decomposition. 

\paragraph{Recommendation systems}
The goal of a recommendation system is to provide personalized items to customers based on their ratings of items, or other information. Usually, the preferences of customers are modeled by an $m\times n$ matrix $\mathbf{A}$, where we assume there are $m$ customers and $n$ items. The element $\mathbf{A}_{lj}$ indicates the level that customer $l$ rates the item $j$. In particular, finding suitable recommendations means to locate the large entries of rows.

One commonly used method for locating the large entries of rows is matrix reconstruction. Given an arbitrary matrix and integer $k$, the goal is to output a small $k$-rank matrix as an approximate matrix. In recommendation system, such an low-rank approximate matrix $\mathbf{A}_{k}$ of the preference matrix $\mathbf{A}$ is usually used to recommend items. In this sense, our VQSVD algorithm finds its application in recommendation system by learning such a small rank matrix. To be more specific, the information about its singular values and vectors can be obtained.

After obtaining the low rank matrix $\mathbf{A}_{k}$, the process of recommendation for customer $l$ proceeds as follows: project the $l$th row vector of $\mathbf{A}$ onto the space spanned by the right singular vectors of $\mathbf{A}_{k}$. 
Specifically, denote the right singular vectors of $\mathbf{A}_{k}$ by $\ket{v_{1}},...,\ket{v_{k}}$, and projection operator by $\Pi\coloneqq \sum_{t=1}^{k}\op{v_{t}}{v_{t}}$, then apply $\Pi$ to the $l$th normalized row vectors $\ket{\mathbf{b}}$ of $\mathbf{A}$. Specially, let $\ket{\widehat{\bf b}}$ denote the normalized vector after projection, given below:
\begin{align}
    \ket{\widehat{\mathbf{b}}}=\frac{\Pi\ket{\mathbf{b}}}{\parallel\Pi\ket{\mathbf{b}}\parallel}=\frac{1}{G}\sum_{t=1}^k\xi_{t}\ket{v_{t}},\label{recommendationprobability}
\end{align}
where the right singular vectors $\{\ket{v_{t}}\}$ form an orthogonal basis, $\xi_{t}$ are the corresponding coefficients, and $G$ denotes the normalization factor, i.e., $G=\sqrt{\sum_{t=1}^k(\xi_{t})^{2}}$. 

The state $\ket{\widehat{\mathbf{b}}}$ can be efficiently prepared via paramerized circuit $V(\bm\beta)$ as long as we can prepare the state $\frac{1}{G}\sum_{t=1}^{k}\xi_{t}\ket{\psi_{t}}$. Particularly, there are $k$ coefficients $\xi_t$ in Eq.~\eqref{recommendationprobability}, determining $\ket{\widehat{\mathbf{b}}}$. 
%
Further, we present a procedure for estimating these coefficients in Algorithm~\ref{alg:qlss} (shown below).
\begin{algorithm}[H] 
\caption{ }
\begin{algorithmic}[1] \label{alg:qlss}
\STATE Input:  matrix $\mathbf{A}\in\mathbb{R}^{n\times n}$, desired rank $k$, parametrized circuits $U(\bm\a)$ and $V(\bm\b)$ with initial parameters of $\bm\a$, $\bm\b$, and tolerance $\varepsilon$, and a circuit $U_{b}$ to prepare state $\ket{\mathbf{b}}\in\mathbb{R}^{n}$;

\STATE Run VQSVD in Algorithm~\ref{alg:rftl} with $\mathbf{A}$, $k$, $U(\bm\a)$, $V(\bm\b)$, $\varepsilon$, and return the optimal parameters $\bm\a^{*}$, $\bm\b^{*}$, and $\{d_{j}\}_{j=1}^{k}$ as singular values;

\FOR {$t=1,\cdots,k$}
\STATE Apply $V(\bm\b)$ to state $\ket{\psi_t}$ and obtain $\ket{v_t} = V(\bm\b^*)\ket{\psi_t}$; 
\STATE Apply $U_{b}$ to state $\ket{\mathbf{0}}$ and obtain $\ket{\mathbf{b}}=U_{b}\ket{\bf 0}$;
\STATE Compute $\xi_{t}=\ip{v_{t}}{\mathbf{b}}$ via Hadamard test;
\ENDFOR 
\STATE Output $\{\xi_{t}\}$ as coefficients in Eq.~\eqref{recommendationprobability}.
\end{algorithmic}
\end{algorithm}
Finally, measure $\ket{\widehat{\mathbf{b}}}$ in the computational basis and output the outcome. 




\paragraph{Polar decomposition and other applications}
Another important application of our VQSVD algorithm is finding the polar decomposition which has many applications in linear algebra~\cite{higham1986computing,gower2004procrustes}. Recently, Lloyd et al~\cite{lloyd2020quantum} proposed a quantum polar decomposition algorithm that performs in a deterministic way. For a given a complex matrix $M\in\mathbb{C}^{n\times n}$, the right polar decomposition is defined as follows,
\begin{align}
    M=WP,
\end{align}
where $W$ is a unitary and $P$ is a positive-semidefinite Hermitian matrix. Particularly, suppose the singular value decomposition of $M$ is calculated through the VQSVD algorithm $M=UD V^{\dagger}$, then $W=UV^{\dagger}$ and $P=VD V^{\dagger}$. It is interesting to explore whether our VQSVD could be applied to polar decomposition. 

For Hermitian matrices, the VQSVD algorithm can be applied as an eigensolver since singular value decomposition reduces to spectral decomposition in this case. 
Recently, some work has been proposed to extract eigenvalues and reconstruct eigenvectors of Hermitian matrices~\cite{Jones2019,nakanishi2019subspace}, density operators~\cite{larose2019variational}. 


In  the  context  of  quantum  information,  SVD  could  be  used  to  compute  the  Schmidt  decomposition of bipartite pure states and we note that our VQSVD algorithm could also be applied to do Schmidt decomposition. Meanwhile, a recent work by Bravo-Prieto et al.~\cite{Bravo-Prieto2019a} introduces a novel varitional quantum algorithm for obtaining the Schmidt coefficients and associated orthonormal vectors of a bipartite pure state. In comparison, our VQSVD algorithm can deal with the SVD of the general matrix, while Ref.~\cite{Bravo-Prieto2019a} is intended for the Schmidt decomposition of  bipartite pure states.

\section{Discussion and outlook}
To summarize, we have presented a variational quantum algorithm for singular value decomposition with NISQ devices. One key contribution is to design a loss function that could be used to train the quantum neural networks to learn the left and right singular vectors and output the target singular values. Further improvements on the performance of our VQSVD algorithm may be done for sparse matrices together with more sophisticated ansatzes. We have numerically verified our algorithm for singular value decomposition of random matrices and image compression and proposed extensive applications in solving linear systems of equations. As a generalization of the family of VQE algorithms on Hamiltonian diagonalization (or spectral decomposition), the VQSVD algorithm may have potential applications in quantum chemistry, quantum machine learning. and quantum optimization in the NISQ era.

\xw{Our algorithm theoretically works well for large-scale problems if we could have carefully handled the trainability and accuracy of QNN. Just like regular variational quantum algorithms, there are still several challenges that need to be addressed to maintain the hope of achieving quantum speedups when scaling up to large-scale problems on NISQ devices. We refer to the recent review papers \cite{Cerezo2020,Endo2020,Bharti2021} for potential solutions to these challenges.
}

One future direction is to develop near-term quantum algorithms for non-negative matrix factorization~\cite{Koren2009} which have various applications and broad interests in machine learning. See \cite{Du2018} as an example of the quantum solution. Another interesting direction is to develop near-term quantum algorithms for higher-order singular value decomposition~\cite{Kolda2009a}.

\section*{Acknowledgements}
We would like to thank Runyao Duan for helpful discussions. This work was done when Z. S. was a research intern at Baidu Research. Y. W. acknowledged support from the Baidu-UTS AI Meets Quantum project and the Australian Research Council (Grant No: DP180100691).
%


\bibliographystyle{unsrtnat}
\bibliography{Reference}

\appendix

\section{Proof details for Proposition~\ref{prop:loss:function:gradient}}
\label{appendix:loss:function:gradient}
In this section, we show a full derivation on the gradients of the loss function in our VQSVD algorithm.
\begin{align}
L(\bm\a,\bm\b) = \sum_{j=1}^T q_j\times \text{Re} \bra{\psi_j}U^{\dagger}(\bm\a) M V(\bm\b)\ket{\psi_j}
\end{align} 
The real part can be estimated via Hadamard test. Equivalently saying that,
\begin{align}
\text{Re} \bra{\psi_j} \hat{O} \ket{\psi_j}
= \frac{1}{2} \bigg[ \bra{\psi_j} \hat{O} \ket{\psi_j} + \bra{\psi_j} \hat{O}^\dagger \ket{\psi_j} \bigg]
\end{align}
Consider the parametrized quantum circuit $U(\bm \a)= \Pi_{i = n}^{1}U_i(\a_i)$ and $V(\bm \b)= \Pi_{j = m}^{1}V_j(\b_j)$. For convenience, denote $U_{i:j} = U_i \cdots U_j$. We can write the cost function as:
\begin{align}
L(\bm\a,\bm\b) = \frac{1}{2} &\sum_{j=1}^T q_j \times  \bigg[\bra{\psi_j}U_{1:n}^{\dagger}(\a_{1:n}) M V_{m:1}(\b_{m:1})\ket{\psi_j} \notag \\
& + \bra{\psi_j}V_{1:m}^{\dagger}(\b_{1:m}) M^\dagger U_{n:1}(\a_{n:1})\ket{\psi_j}\bigg]
\end{align} 
Absorb most gates into state $\ket{\psi_j}$ and matrix $M$,
\begin{align}
L(\bm\a,\bm\b) = \frac{1}{2} &\sum_{j=1}^T q_j \times  \bigg[\bra{\phi_j}U_{\ell}^{\dagger}(\a_\ell) G V_{k}(\b_k)\ket{\varphi_j} \notag \\
& + \bra{\varphi_j}V^{\dagger}_{k}(\b_k) G^\dagger U_{\ell}(\a_\ell)\ket{\phi_j}\bigg]
\end{align} 
where $\ket{\varphi_j} = V_{k-1:1}\ket{\psi_j}$, $\ket{\phi_j} = U_{\ell-1:1}\ket{\psi_j}$ and $G \equiv U^\dagger_{\ell+1:n} M V_{m:k+1}$.
We assume $U_\ell=e^{-i\a_\ell H_\ell/2}$ is generated by a Pauli product $H_\ell$ and same for $V_k = e^{-i\b_k Q_k/2}$. The derivative with respect to a certain angle is
\begin{align}
\frac{\partial U_{1:n}}{\partial \a_\ell} &= -\frac{i}{2}U_{1:\ell-1} (H_\ell U_\ell) U_{\ell+1:n}\\
\frac{\partial V_{1:m}}{\partial \b_k} &= -\frac{i}{2}V_{1:k-1} (Q_k V_k) V_{k+1:m}
\end{align}
Thus the gradient is calculated to be
\begin{align}
\frac{\partial L(\bm\a,\bm\b)}{\partial \a_\ell} = \frac{1}{2} &\sum_{j=1}^T q_j \times  \bigg[\frac{i}{2}\bra{\phi_j} H^\dagger_\ell U_{\ell}^{\dagger}(\a_\ell) G V_{k}(\b_k)\ket{\varphi_j} \notag \\
& -\frac{i}{2} \bra{\varphi_j}V^{\dagger}_{k}(\b_k) G^\dagger H_\ell U_{\ell}(\a_\ell)\ket{\phi_j}\bigg]
\end{align}
With the following property, we can absorb the Pauli product $H_\ell$ by an rotation on $\a_\ell\rightarrow \a_\ell +\pi$
\begin{align}
U_\ell(\pm \pi) & = e^{\mp i\pi H_\ell/2} \notag \\
& = \cos(\frac{\pi}{2}) I \mp i \sin(\frac{\pi}{2})H_\ell \notag \\
& = \mp i H_\ell
\end{align}
Plug it back and we get,
\begin{align}
\frac{\partial L(\bm\a,\bm\b)}{\partial \a_\ell} = \frac{1}{4} &\sum_{j=1}^T q_j \times  \bigg[\bra{\phi_j}  U_{\ell}^{\dagger}(\a_\ell + \pi) G V_{k}(\b_k)\ket{\varphi_j} \notag\\
& + \bra{\varphi_j}V^{\dagger}_{k}(\b_k) G^\dagger  U_{\ell}(\a_\ell+\pi)\ket{\phi_j}\bigg]
\end{align}
This can be further simplified as
\begin{align}
\frac{\partial L(\bm\a,\bm\b)}{\partial \a_\ell} & = \frac{1}{2} \sum_{j=1}^T q_j\times \text{Re} \bra{\psi_j}U^{\dagger}(\a_{\ell}+\pi) M V(\bm\b)\ket{\psi_j} \notag \\
& = \frac{1}{2} L(\a_\ell + \pi,\bm\b)
\end{align}
Similarly, for $\b_k$ we have
\begin{align}
\frac{\partial L(\bm\a,\bm\b)}{\partial \b_k} = \frac{1}{2}& \sum_{j=1}^T q_j \times  \bigg[\frac{i}{2}\bra{\phi_j}  U_{\ell}^{\dagger}(\a_\ell) G  Q_k V_{k}(\b_k)\ket{\varphi_j} \notag \\
& -\frac{i}{2} \bra{\varphi_j} Q^\dagger_k V^{\dagger}_{k}(\b_k) G^\dagger U_{\ell}(\a_\ell)\ket{\phi_j}\bigg]
\end{align}
This can be further simplified as
\begin{align}
\frac{\partial L(\bm\a,\bm\b)}{\partial \b_k} &= \frac{1}{2} \sum_{j=1}^T q_j\times \text{Re} \bra{\psi_j}U^{\dagger}(\bm \a) M V(\b_k-\pi)\ket{\psi_j} \notag \\
& = \frac{1}{2} L(\bm \a,\b_k - \pi).
\end{align}
One can see from the above derivation that calculating the analytical gradient of VQSVD algorithm simply means rotating a specific gate parameter (angle $\a_\ell$ or $\b_k$) by $\pm \pi$ which can be easily implemented on near-term quantum devices.


\section{Supplemental material for cost evaluation}
\label{sec:supplemental:cost_evaluation}
\xw{In this section, we show how to apply the importance sampling technique to reduce the cost of estimating the loss function in Eq.~\eqref{eq:decom of M value}. For convenience, we restate Eq.~\eqref{eq:decom of M value} below.
\begin{align}\label{eq:cost_sampling}
    C(\bm\alpha,\bm\beta)=\sum_{k=1}^K  c_k \times \text{Re}\bra{\psi_j}U(\bm\a)^{\dagger} A_k V(\bm\b)\ket{\psi_j}.
\end{align}
We assume all coefficients $c_{k}$ in Eq.~\eqref{eq:decom of M value} are positive, since their signs can be absorbed into unitaries $A_k$. 
}

\xw{
Define an importance sampling in a sense that
\begin{equation}
    \mathbf{R}=
    \left\{
    \begin{array}{ll}
    A_{1} & \text{with probability $p_1=\frac{c_1}{\|\mathbf{c}\|_{\ell_1}}$}\\
    A_{2} & \text{with probability $p_2=\frac{c_2}{\|\mathbf{c}\|_{\ell_1}}$}\\
    \vdots & \vdots\\
    A_{K} & \text{with probability $p_1=\frac{c_K}{\|\mathbf{c}\|_{\ell_1}}$}
    \end{array}
    \right.
\end{equation}
where $\mathbf{c}=(c_1,c_2,\ldots,c_K)$, and $\|\cdot\|_{\ell_1}$ denotes the $\ell_1$-norm. Random variable $\mathbf{R}$ indicates that each unitary $A_k$ is selected at random with probability proportional to its weight. Here, we rewrite Eq.~\eqref{eq:cost_sampling} by using $\mathbf{R}$.
\begin{align}
    C(\bm\alpha,\bm\beta)=\mathbf{E}\left[\|\mathbf{c}\|_{\ell_1}\cdot{\rm Re}\bra{\psi_j}\mathbf{R}V(\bm\beta)\ket{\psi_j}\right].
\end{align}
}

\xw{
As $C(\bm\alpha,\bm\beta)$ is formed as an expectation, the sample mean is supposed to estimate it. By Chebyshev's inequality, $O(\mathbf{Var}/\epsilon^2)$ samples suffice to compute an estimate up to precision $\epsilon$ with high probability, where $\mathbf{Var}$ denotes the variance, and $\epsilon$ denotes the precision. By Chernoff bound, the probability can be further improved to $1-\delta$ costing an additional multiplicative factor $O(\log(1/\delta))$. Regarding $C(\bm\alpha,\bm\beta)$, the variance is bounded by $\|\mathbf{c}\|_{\ell_1}^2$. We, therefore, only need to choose $O(\|\mathbf{c}\|_{\ell_1}^2\log(1/\delta)/\epsilon^2)$ many unitaries to evaluate $C(\bm\alpha,\bm\beta)$. Viewed from this point, the cost of the evaluation is dependent on the $\|\mathbf{c}\|_{\ell_1}$ instead of the integer $K$. As a result, the loss evaluation could be efficient if the $\|\mathbf{c}\|_{\ell_1}$ is small, even there are exponentially many terms. 
}

\xw{
\textbf{Example.} Suppose we use VQSVD to compute SVDs of circulant matrices, which have applications in signal processing, image compression, number theory and cryptography. A circulant matrix $\mathbf{C}_{d}$ is given by
\begin{align}
    \mathbf{C}_{d}=
    \begin{bmatrix}
    c_0 & c_1 & \ldots & c_{d-1} \\
    c_{d-1} & c_0 & \ldots & c_{d-2}\\
    \vdots & \vdots & \ddots & \vdots \\
    c_{1} & c_{2} & \ldots & c_{0}
    \end{bmatrix}
\end{align}
Clearly, the matrix $\mathbf{C}_{d}$ is determined by the sequence $\mathbf{c}=(c_0,c_1,\ldots,c_{d-1})$. It can be easily decomposed into a weighted sum of cyclic permutations. 
\begin{align}
    \mathbf{C}_{d}=c_0I+c_{1}P_1+c_2P_2+\ldots+c_{d-1}P_{d-1},
\end{align}
where 
\begin{align}
    P_1=
    \begin{bmatrix}
    0 & 0 & \ldots & 0 & 1\\
    1 & 0 & \ldots & 0 & 0\\
    0 & 1 & \ldots & 0 & 0\\
    \vdots &  & \ddots &  & \vdots\\
    0 & 0 & \ldots & 1 & 0
    \end{bmatrix},P_2=
    \begin{bmatrix}
    0 & 0 & \ldots & 1 & 0\\
    0 & 0 & \ldots & 0 & 1\\
    1 & 0 & \ldots & 0 & 0\\
    \vdots & \ddots & 0 &  & \vdots\\
    0 & \ldots & 1 & 0 & 0
    \end{bmatrix},& \nonumber\\
    \ldots,P_{d-1}=
    \begin{bmatrix}
    0 & 1 & 0 & \ldots & 0\\
    0 & 0 & 1 & \ldots & 0\\
    0 & 0 & 0 & \ldots & 0\\
    \vdots &  &  & \ddots & 1\\
    1 & 0 & \ldots & 0 & 0
    \end{bmatrix}.& \nonumber
\end{align}
The loss for a circulant matrix in Eq.~\eqref{eq:cost_sampling} is given by
\begin{align}
    &\sum_{k=0}^{d-1}c_{k}\times {\rm Re}\bra{\psi_j}U(\bm\alpha)^\dagger P_{k}V(\bm\beta)\ket{\psi_j}\nonumber\\
    &=\mathbf{E}\left[\|\mathbf{c}\|_{\ell_1}\cdot{\rm Re}\bra{\psi_j}U(\bm\alpha)^\dagger P_{k}V(\bm\beta)\ket{\psi_j}\right]
\end{align}
Hence, for circulant matrices with feasible cost $\|\mathbf{c}\|_{\ell_1}$ (e.g., polynomial size), the loss evaluation could be quite efficient even for large-size problems. Our VQSVD algorithm is expected to find further applications in circulant matrix-related problems.
}

\section{Supplemental material for different ansatzs}
\label{sec:supplemental:numerics}
In this section, we supplement additional numerics on $8\times 8$ real matrices to explore the effect of different circuit ansatzs and later test our VQSVD algorithm on random complex matrices. First, we introduce several ansatz candidates.

\begin{figure}[h]
\centering
\text{(a)}\, \Qcircuit @C=0.7em @R=0.4em{
&\gate{U}&\ctrl{1}&\qw&\qw\\ 
&\gate{U}&\targ&\ctrl{1}&\qw\\
&\gate{U}&\qw&\targ&\qw^{\quad \quad \quad \,\times D_1}  
\gategroup{1}{2}{3}{4}{2.07em}{--}
}
\quad\text{(b)}\,\Qcircuit @C=0.7em @R=0.4em{
&\gate{U}&\ctrl{1}&\gate{U}   &\qw    &\qw      &\qw&\qw\\ 
&\gate{U}&\targ   &\gate{U}   &\gate{U}&\ctrl{1}&\gate{U}&\qw\\
&\qw     &\qw     &\qw        &\gate{U}&\targ   &\gate{U}&\qw^{\quad \quad \quad \quad \,\times D_2} 
\gategroup{1}{2}{3}{7}{2.07em}{--}
}
\vspace{2mm}
\\
\text{(c)}\,\Qcircuit @C=0.7em @R=0.4em{
&\gate{U} &\ctrl{1}&\qw     &\targ&\qw\\ 
&\gate{U} &\targ   &\ctrl{1}&\qw&\qw\\
&\gate{U} &\qw     &\targ   &\ctrl{-2} &\qw^{\quad \quad \quad \,\times D_3} 
\gategroup{1}{2}{3}{5}{2.07em}{--}
}
\quad\text{(d)}\, \Qcircuit @C=0.7em @R=0.4em{
&\gate{U} &\ctrl{1}&\qw       &\gate{U} &\targ &\qw&\qw\\ 
&\gate{U} &\targ   &\ctrl{1}  &\gate{U} &\qw  &\targ&\qw\\
&\gate{U} &\qw     &\targ     &\gate{U} &\ctrl{-2} &\ctrl{-1}&\qw^{\quad \quad \quad \,\times D_4} 
\gategroup{1}{2}{3}{7}{2.07em}{--}
}
\caption{Ansatz candidates considered in numerical experiments. When the input data $M$ contains only real elements, we take $U = R_y(\alpha_j)$ to reduce the trainable parameters. Otherwise, we choose $U = R_z(\theta_j)R_y(\phi_j)R_z(\varphi_j)$ as a general rotation on the Bloch sphere to enhance expressibility. (a) is the basic hardware-efficient ansatz. (b) consists of dressed CNOT gates (rotation gates on both sides). (c) and (a) differs by an extra CNOT which forms a circular entangling structure and is expected to have a stronger entangling capability. (d) is designed to explore the influence of CNOT gates.}
\label{fig:AnsatzFamily}
\end{figure}
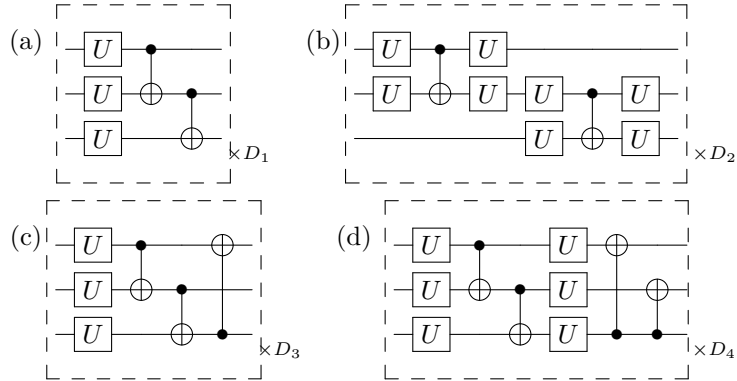

After introducing all the candidates, we test the distance measure on a real matrix similar to Fig.~\ref{fig:VQSVD_performance}. For fair comparison, we want all the ansatz candidates to have the same total amount of trainable parameters $N_{tot}=24$ and leads to the specific depth $D_1 = 8, D_2 = 3, D_3 = 8, D_4 = 4$. Other setups including desired rank $T=8$, the weight coefficient $(T,T-1,\cdots,1)$, the random seed for generating trainable parameters, number of optimization iterations ITR=200, learning rate LR=0.05, and the Adam optimizer are fixed. The result is illustrated in Fig.~\ref{fig:AnsatzFamily_Comparison}.
The advantage of candidate (a) and (d) is observed consistently in numerical experiments. Finally, we repeat the same experiments on random $8\times 8$ complex matrices. No clear advantage of a specific ansatz candidate is observed under this setup. Most of the time, candidate (a), (c) and (d) would return a better result compare to candidate (b). But in certain cases, the performance of candidate (a) could be very bad due to the lack of entangling capability. We report a case study in Fig.~\ref{fig:AnsatzFamily_Complex}. Further studies are needed to check more ansatz candidates by manipulating the entangling structure.

 \begin{figure}[t]
 \centering
\includegraphics[width=0.7\columnwidth]{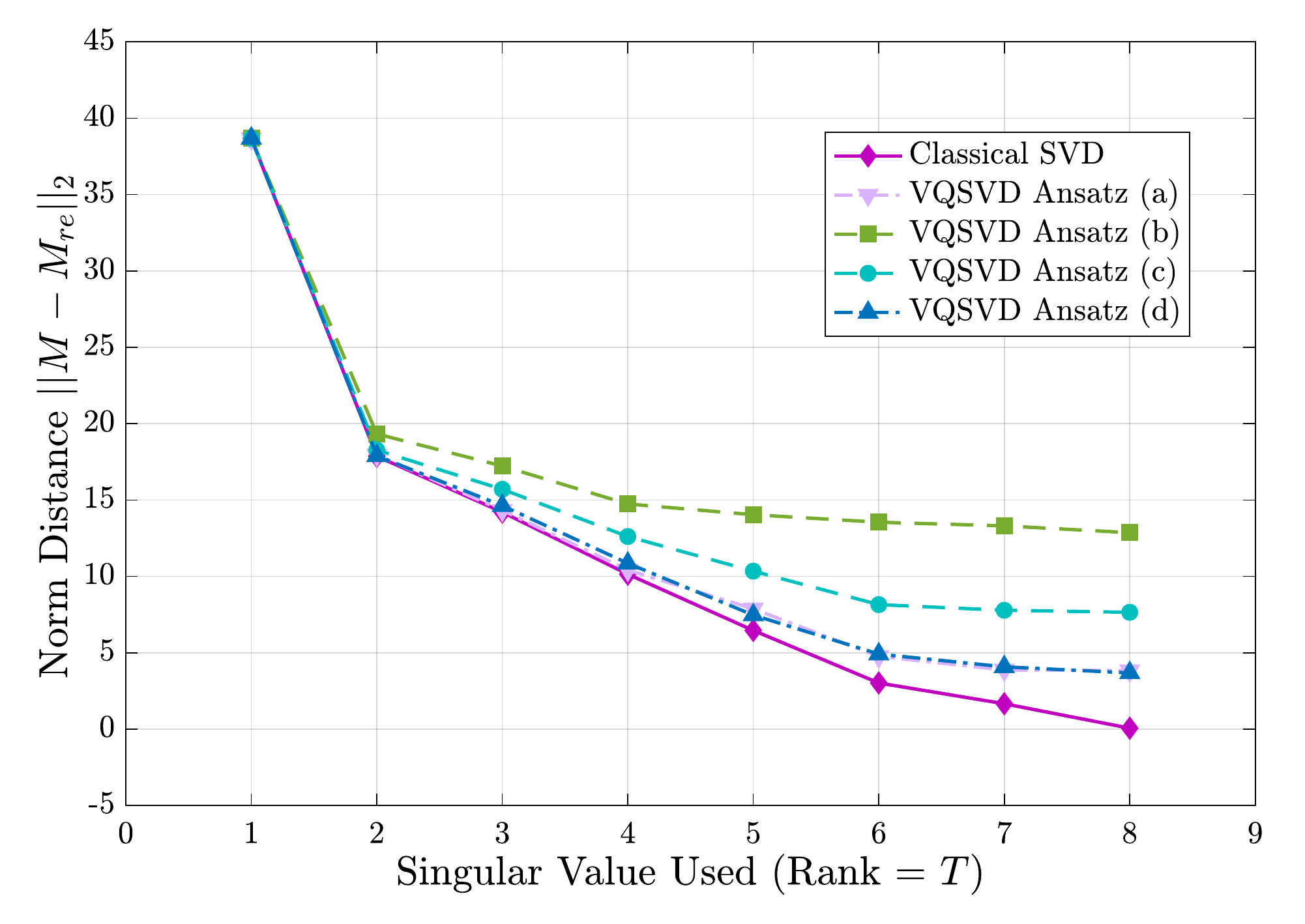}
\caption{\label{fig:AnsatzFamily_Comparison} Distance measure between the reconstructed matrix $M_{re}$ and the original matrix $M$ (real elements only) via VQSVD with different ansatz candidate and compare with the classical SVD method. }
\end{figure}

 \begin{figure}[t]
 \centering
\includegraphics[width=0.7\columnwidth]{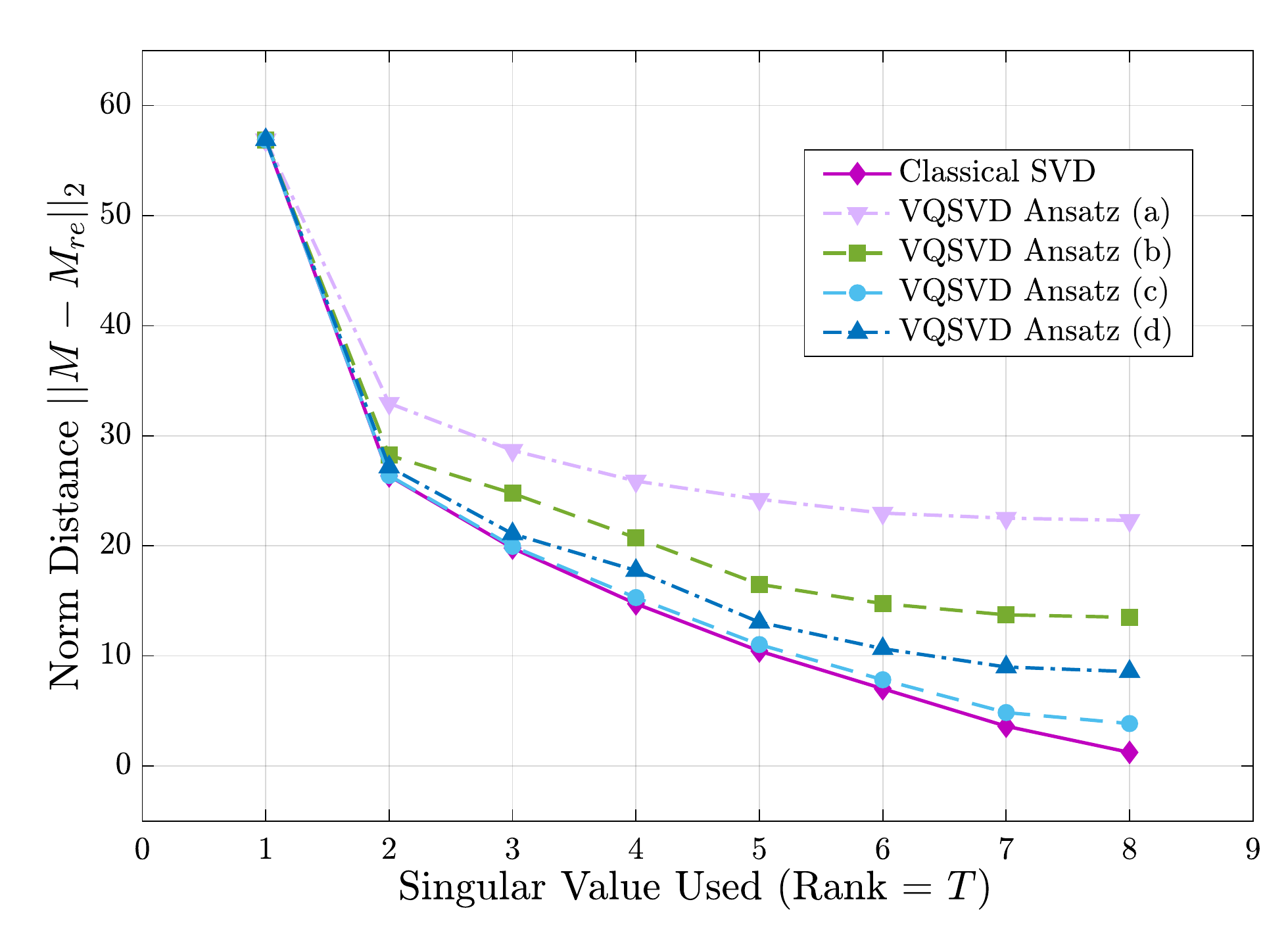}
\caption{\label{fig:AnsatzFamily_Complex} Distance measure between the reconstructed matrix $M_{re}$ and the original matrix $M$ (complex elements) via VQSVD with different ansatz candidate and compare with the classical SVD method. The number of trainable parameters is now changed to $N_{tot}=144$.}
\end{figure}

\section{Supplemental material for verification of the solution quality}\label{sec:supplemental:verification}
In this section, we provide the necessary proofs in Sec.~\ref{sec:verification} and detailed discussions on variational quantum Frobenius norm estimation algorithm. 
\subsection{Definitions}
Recall that the error of inferred singular values is defined as follows,
\begin{align}
\epsilon_{d}=\sum_{j=1}^{T}(d_{j}-m_{j})^{2},
\label{error:singular:values}
\end{align}
where $d_{j}$ are the exact singular values of matrix $M$ and also arranged in descending order. And the error $\epsilon_{v}$ of inferred singular vectors is defined below,
\begin{align}
\epsilon_{v}=\sum_{j=1}^{T} \parallel H\ket{\hat{e}_{j}^{+}}-m_{j}\ket{\hat{e}_{j}^{+}}\parallel^{2}+\sum_{j=1}^{T} \parallel H\ket{\hat{e}_{j}^{-}}+m_{j}\ket{\hat{e}_{j}^{-}}\parallel^{2},
\label{error:singular:vectors}
\end{align}
where $H$ is a Hermitian of the form $H=\op{0}{1}\otimes M+\op{1}{0}\otimes M^{\dagger}$, and $\ket{\hat{e}_{j}^{\pm}}=(\ket{0}\ket{\hat{u}_{j}}\pm\ket{1}\ket{\hat{v}_{j}})/\sqrt{2}$. 
The quantity $\|H\ket{\hat{e}_{j}^{\pm}}\mp m_{j}\ket{\hat{e}_{j}^{\pm}}\|^{2}$ quantifies the component of $H\ket{\hat{e}_{j}^{\pm}}$ that is perpendicular to $\ket{\hat{e}_{j}^{\pm}}$, which follows from $(I-\op{\hat{e}_{j}^{\pm}}{\hat{e}_{j}^{\pm}})H\ket{\hat{e}_{j}^{\pm}}$.

It is worth pointing out that when inferred vectors $\ket{\hat{e}_{j}^{\pm}}$ approximate the eigenvectors $\ket{e_{j}^{\pm}}$, where $\ket{e_{j}^{\pm}}=(\ket{0}\ket{u_j}\pm\ket{1}\ket{v_j})/\sqrt{2}$, of $H$, i.e., $\epsilon_{v}\to0$, inferred singular vectors $\ket{\hat{u}_{j}}$ and $\ket{\hat{v}_{j}}$ approximate the singular vectors $\ket{u_{j}}$ and $\ket{v_{j}}$ respectively, and vice versa. On the other hand, the error $\epsilon_{v}$ which is used to quantify the extent that vectors $\ket{\hat{e}_{j}^{\pm}}$ approximate eigenvector $\ket{e_{j}^{\pm}}$ can quantify the extent that inferred vectors $\{u_j\}$ and $\{v_j\}$ approximate the singular vectors. Specifically, these distances have an equal relation, which is depicted in the following equation.
\begin{align}
D(\{\ket{u_{j}},\ket{v_{j}}\},\{\ket{\hat{u}_{j}},\ket{\hat{v}_{j}}\})=D(\{\ket{e_{j}^{+}},\ket{e_{j}^{-}}\},\{\ket{\hat{e}_{j}^{+}},\ket{\hat{e}_{j}^{-}}\}),
\label{distance:equality}
\end{align}
where $D$ denotes the distance between vectors.

Here, we give the explicit forms of the distances .~\eqref{distance:equality}. The distances between $\{\ket{u_{j}},\ket{v_{j}}\}$ and $\{\ket{\hat{u}_{j}},\ket{\hat{v}_{j}}\}$ are defined in the following form,
\begin{align}
&D(\{\ket{u_{j}},\ket{v_{j}}\},\{\ket{\hat{u}_{j}},\ket{\hat{v}_{j}}\})\\
\equiv&\parallel \ket{u_{j}}-\ket{\hat{u}_{j}}\parallel^{2}+\parallel \ket{v_{j}}-\ket{\hat{v}_{j}}\parallel^{2}.
\label{eq:distance:uv}
\end{align}
And the distances between $\ket{e_{j}^{\pm}}$ and $\ket{\hat{e}_{j}^{\pm}}$ are defined below,
\begin{align}
&D(\{\ket{e_{j}^{+}},\ket{e_{j}^{-}}\},\{\ket{\hat{e}_{j}^{+}},\ket{\hat{e}_{j}^{-}}\})\\
\equiv&\parallel\ket{{e}_{j}^{+}}-\ket{\hat{e}_{j}^{+}}\parallel^{2}+\parallel\ket{e_{j}^{-}}-\ket{\hat{e}_{j}^{-}}\parallel^{2}.
\label{eq:distance:e}
\end{align}
Notice that both of the Right-Hand-Sides of Eqs.~\eqref{eq:distance:uv},~\eqref{eq:distance:e} are equivalent to $4-2({\rm Re}\ip{u_{j}}{\hat{u}_{j}}+{\rm Re}\ip{v_{j}}{\hat{v}_{j}})$, and then the relation in Eq.~\eqref{distance:equality} follows.

\subsection{Error analysis}
After running VQSVD, it would be ideal if we can verify the quality of the outputs from VQSVD. For achieving this purpose, we show that these error $\epsilon_{d}$ and $\epsilon_{v}$ are upper bounded and give the explicit form of upper bounds. We present the derivation for upper bounds on errors in the following lemma.
\begin{lemma}\label{le:error:upper:bound}
Given a matrix $M\in\mathbb{C}^{n\times n}$, let $\epsilon_{d}$ and $\epsilon_{v}$ denote the errors of the inferred singular values and singular vectors in Eqs.~\eqref{error:singular:values},~\eqref{error:singular:vectors}, respectively, then both of them are upper bounded. To be more specific,
\begin{align*}
\epsilon_{d}&\leq\sum_{j=1}^{T}d_{j}^{2}-\sum_{j=1}^{2}m_{j}^{2},\\
\epsilon_{v}&\leq 2(\sum_{j=1}^{T}d_{j}^{2}-\sum_{j=1}^{T}m_{j}^{2}),
\end{align*}
where $d_{j}$s are singular values of matrix $M$, and $m_{j}$s are inferred singular values from Algorithm~\ref{alg:rftl}.
\end{lemma}
\begin{proof}
Recall the definitions of $\epsilon_{d}$ and $\epsilon_{v}$ in Eqs.~\eqref{error:singular:values},~\eqref{error:singular:vectors}, and notice that
\begin{align}
\epsilon_{d}=\sum_{j=1}^{T}d_{j}^{2}-2\sum_{j=1}^{T}d_{j}m_{j}+\sum_{j=1}^{T}m_{j}^{2}.
\end{align}
Since the dot product with decreasingly ordered coefficients is Schur-convex and $\{d_{j}\}$ majorize $\{m_{j}\}$, i.e., $\sum_{j=1}^{\ell}d_{j}\geq\sum_{j=1}^{\ell}m_{j}$ for all $\ell=1,...,T$, then $\sum_{j=1}^{T}d_{j}m_{j}\geq\sum_{j=1}^{T}m_{j}^{2}$, which results an upper bound on error $\epsilon_{d}$. 

Note that the error $\epsilon_{v}$ can be rewritten as 
\begin{align}
\epsilon_{v}=\sum_{j=1}^{T}(\bra{\hat{e}_{j}^{+}}H^{2}\ket{\hat{e}_{j}^{+}}+\bra{\hat{e}_{j}^{-}}H^{2}\ket{\hat{e}_{j}^{-}}-2\sum_{j=1}^{T}m_{j}^{2},
\end{align}
and eigenvectors $\{\ket{\hat{e}_{j}^{\pm}}\}$ can be expanded into a basis of the space, then we have
\begin{align}
&\sum_{j=1}^{T}(\bra{\hat{e}_{j}^{+}}H^{2}\ket{\hat{e}_{j}^{+}}+\bra{\hat{e}_{j}^{-}}H^{2}\ket{\hat{e}_{j}^{-}}\\
\leq&\tr(H^{2})= 2\sum_{j=1}^{T}d_{j}^{2}.
\end{align}
\end{proof}

\subsection{A tool for solution quality estimation}
The goal of this section is to provide a tool that computes the sum of largest $T$ squared singular values, which is used in Sec.~\ref{sec:verification} to analyze the accuracy of outputs of VQSVD. In the following, we mainly show the correctness analysis of Algorithm~\ref{algorithm:vqfne}. The detailed discussions on loss evaluation and gradients estimation are omitted, since we can employ the same methods introduced in Ref.~\cite{Bravo-Prieto2019} to loss evaluation and gradients derivation. The differences of our method from that in Ref.~\cite{Bravo-Prieto2019} occur in input states. Specifically, we input computational states $\ket{\psi_{j}}$ for all $j$ into the circuit, while they input state $\ket{\bf0}$. For more information on loss evaluation and gradients derivation, we refer the interested readers to Ref.~\cite{Bravo-Prieto2019}. 
\begin{algorithm}[H] 
\caption{ }
\label{algorithm:vqfne}
\begin{algorithmic}[1] 
\STATE Input:  $\{c_k, A_k\}_{k=1}^K$, desired rank $T$, parametrized circuits $U(\bm\a)$ and $V(\bm\b)$ with initial parameters of $\bm\a$,  $\bm\b$, and tolerance $\varepsilon$;

\STATE Choose computational basis $\ket{\psi_1},\cdots,\ket{\psi_T}$;

\FOR {$j=1,\cdots,T$}
 \STATE Apply $U(\bm\a)$ to state $\ket{\psi_j}$ and obtain $\ket {u_j} = U(\bm\a)\ket{\psi_j}$; 
 
\STATE Apply $V(\bm\b)$ to state $\ket{\psi_j}$ and obtain $\ket {v_j} = V(\bm\b)\ket{\psi_j}$ ;

\STATE Compute $o_{j}=|\bra{u_j} M \ket{v_j}|^{2}$ via Hadamard test;
\ENDFOR 

\STATE Compute the loss function $F(\bm\a,\bm\b) =  \sum_{j=1}^T o_j$;

\STATE Perform optimization to maximize $F(\bm\a,\bm\b)$, update parameters of $\bm\a$ and $\bm\b$;

\STATE Repeat 4-10 until the loss function $F(\bm\a,\bm\b)$ converges with tolerance $\varepsilon$;
\STATE Output $F(\bm\a,\bm\b)$ as Frobenius norm.
\end{algorithmic}
\end{algorithm}

\paragraph{Correctness analysis}
The validity of Algorithm~\ref{algorithm:vqfne} follows from a fact that, for arbitrary matrix, its squared singular values majorize the squared norms of diagonal elements. Specifically, the sum of the largest $T$ squared singular values is larger than the sum of squared norms of the largest $T$ diagonal elements. We summarize this fact in the lemma below and further provide a proof.
\begin{lemma}
For arbitrary matrix $M\in\mathbb{C}^{N\times N}$, let singular values of $M$ be $d_1,d_2,...,d_N$, which are arranged in descending order. Then for any $k\in[N]$, we have the following inequality:
\begin{align}
    \sum_{j=1}^{k}d_{j}^2\geq\sum_{j=1}^{k}|D_{j}^\downarrow|^2,
    \label{singular:values:majorization1}
\end{align}
where $D$ is the diagonal vector of $M$, i.e., $D={\rm diag}(M)$, and the notation $\downarrow$ means that the $|D_j|$ are arranged in descending order. In particular, the equality holds if and only if $M$ is diagonal.
\end{lemma}

\begin{proof}
The core of the proof is to connect singular values and diagonal elements. Specifically, this process can been done with
the following inequalities:
\begin{align}
&\sum_{j=1}^{k}d_{j}^{2}\geq \sum_{j=1}^{k}\vec{M}_{j}^{\downarrow},
\label{revise:inequality:singular:value:diagonal:element:1}\\
&\sum_{j=1}^{k}\vec{M}_{j}^{\downarrow}\geq\sum_{j=1}^{k}|D_{j}^{\downarrow}|^{2},\label{revise:inequality:singular:value:diagonal:element:2}
\end{align}
where $d_j$ are singular values of $M$, 
and $\vec{M}$ is the vector  denoting the diagonal elements of $MM^{\dagger}$. The first inequality Eqs.~\eqref{revise:inequality:singular:value:diagonal:element:1} can be derived since eigenvalues of a Hermitian matrix majorize its diagonal elements. As $d_{j}^{2}$s are the eigenvalues of $MM^{\dagger}$, we have  
\begin{align}
\sum_{j=1}^{k}d_{j}^{2}\geq \sum_{j=1}^{k}\vec{M}_{j}^{\downarrow},
\end{align}
for any $k\in[N]$. Second, note that diagonal elements of $MM^{\dagger}$, i.e., $\vec{M}_{j}$, can be expressed in the following form:
\begin{align}
\vec{M}_{j}=\sum_{l=1}^{N}|M_{jl}|^{2}.
\label{diagonal:element:inequality}
\end{align}
Then, from Eq.~\eqref{diagonal:element:inequality}, we can easily derive inequalities below:
\begin{align}
\sum_{j=1}^{k}\vec{M}_{j}^{\downarrow}\geq \sum_{l\in S}\vec{M}_{l} \geq\sum_{j=1}^{k}|D_{j}^{\downarrow}|^{2}.\label{intermediate_inequality}
\end{align}
where sum $\sum_{j}^{k}|D_{j}^{\downarrow}|^{2}$ includes the largest $k$ absolute values of diagonal elements $D_j$, and set $S$ contains indices of rows that all $|D_j^\downarrow|$ belong.
Here we explain inequalities in Eq.~\eqref{intermediate_inequality}. The second inequality holds since the sum $\sum_{l\in S}\vec{M}_{l}$ contains entries of rows that $|D_j|^\downarrow$ locate. Thus it not only contains the diagonal elements that appear in the sum $\sum_{j}^{k}|D_{j}^{\downarrow}|^{2}$ but also other off-diagonal elements. Meanwhile, recall that the sum $\sum_{j=1}^{k}\vec{M}_{j}^{\downarrow}$ consists of the largest $k$ diagonal elements of matrix $MM^\dagger$. Thus it must be larger than any sum $\sum_{l\in S}\vec{M}_l$ that contains some $k$ diagonal elements of $MM^\dagger$, validating the first inequality.  


Note that the equality in Eq.~\eqref{revise:inequality:singular:value:diagonal:element:2} holds only when $\vec{M}_{j}^{\downarrow}=|D_j^\downarrow|^2$ for all $j\in[N]$. Thus, it implies that $M$ is diagonal. On the other hand, if the matrix $M$ is diagonal, then the equality in Eq.~\eqref{revise:inequality:singular:value:diagonal:element:1} immediately follows. Overall, the equality in Eq.~\eqref{singular:values:majorization1} holds if and only if $M$ is diagonal.
\end{proof}

\paragraph{Loss evaluation}
We consider the evaluation of $o_{j}$ in VQFNE, which can be rewritten as 
\begin{align}
    o_{j}&=\bra{u_{j}}M\ket{v_{j}}\bra{v_{j}}M^{\dagger}\ket{u_{j}}\\
    &=\sum_{k_{1},k_{2}}c_{k_{1}}c_{k_{2}}\bra{u_{j}}A_{k_{1}}\ket{v_{j}}\bra{v_{j}}A_{k_{2}}^{\dagger}\ket{u_{j}}.
    \label{eq:VQFNE:oj}
\end{align}
In principle, these inner products in Eq.~\eqref{eq:VQFNE:oj} can be efficiently estimated via Hadamard test and a little classical post-processing. Actually, there are other methods named Hadamard-overlap test for estimating $o_{j}$. Hadamard-overlap test was introduced in Ref.~\cite{Bravo-Prieto2019} to compute a quantity of the form $\bra{\bf0}U^{\dagger}A_{l}V\ket{\bf0}\bra{\bf0}V^{\dagger}A_{l'}^{\dagger}U\ket{\bf0}$, while in VQFNE, we substitute state $\ket{\bf0}$ with state $\ket{\psi_j}$, which makes no difference in loss evaluation and gradients derivation. Particularly, instead of estimating each inner product $\bra{u_{j}}A_{k_{1}}\ket{v_{j}}$ and $\bra{v_{j}}A_{k_{2}}^{\dagger}\ket{u_{j}}$, the values $\bra{u_{j}}A_{k_{1}}\ket{v_{j}}\bra{v_{j}}A_{k_{2}}^{\dagger}\ket{u_{j}}$ can be estimated via Hadamard-overlap test at the expense of doubling the number of qubits. 
\paragraph{Gradients}
The gradient of the loss function $F(\bm\a,\bm\b)$ is given below,
\begin{align}
    \nabla F(\bm\a,\bm\b)=(\frac{\partial F}{\partial\alpha_{1}},...,\frac{\partial F}{\partial \alpha_{h_{1}}},\frac{\partial F}{\partial\beta_{1}},...,\frac{\partial F}{\partial\beta_{h_{2}}}),
\end{align}
where
\begin{align}
    \frac{\partial F}{\partial\alpha_{l}}&=\sum_{j}\frac{\partial o_{j}}{\partial_{\alpha_{l}}}=\sum_{j}\sum_{k_{1}k_{2}}c_{k_{1}}c_{k_{2}}\frac{\partial R_{j,k_1,k_2}}{\partial\alpha_{l}}\label{eq:VQFNE:gradietns:alpha}\\
    \frac{\partial F}{\partial\beta_{t}}&=\sum_{j}\frac{\partial o_{j}}{\partial_{\beta_{t}}}=\sum_{j}\sum_{k_{1}k_{2}}c_{k_{1}}c_{k_{2}}\frac{\partial R_{j,k_1,k_2}}{\partial\beta_{t}}.\label{eq:VQFNE:gradietns:beta}
\end{align}
where $R_{j,k_1,k_2}=\bra{u_{j}}A_{k_{1}}\ket{v_{j}}\bra{v_{j}}A_{k_{2}}^{\dagger}\ket{u_{j}}$.

More details on deriving gradients in Eqs.~\eqref{eq:VQFNE:gradietns:alpha},~\eqref{eq:VQFNE:gradietns:beta} can be found in Ref.~\cite{Bravo-Prieto2019}.

\end{document}